\documentclass[11pt]{llncs}
\usepackage{amsmath,amssymb}
\usepackage{graphicx}
\newcommand{\comm}[1]{}

\usepackage{fullpage}
\usepackage[top=0.8in, bottom=0.9in, left=0.9in, right=0.9in]{geometry}

\usepackage{color,soul}
\usepackage{amsmath,amssymb}	
\usepackage[euler]{textgreek}
\usepackage{algorithm}
\usepackage{algpseudocode}
\usepackage{cite}
\usepackage{enumitem}

\usepackage[pdftex, plainpages = false, pdfpagelabels, 
                 bookmarks=false,
                 bookmarksopen = true,
                 bookmarksnumbered = true,
                 breaklinks = true,
                 linktocpage,
                 pagebackref,
                 colorlinks = true,  
                 linkcolor = blue,
                 urlcolor  = cyan,
                 citecolor = red,
                 anchorcolor = green,
                 hyperindex = true,
                 hyperfigures
                 ]{hyperref}

\def\calI{\mathcal{I}}

\def\calD{\mathcal{D}}
\def\calL{\mathcal{L}}
\def\hs{\hat{s}}
\def\hp{\hat{p}}

\def\tP{\Tilde{P}}
\def\tS{\Tilde{S}}

\def\vd{V\!D}

\def\optp{P_{\text{opt}}}
\newtheorem{observation}{Observation}

\begin{document}

\title{Unweighted Geometric Hitting Set for Line-Constrained Disks and Related Problems\thanks{A preliminary version of this paper will appear in {\em Proceedings of the 49th International Symposium on Mathematical Foundations of Computer Science (MFCS 2024)}. This research was supported in part by NSF under Grant CCF-2300356.}}
\author{Gang Liu 
\and
Haitao Wang 
}

 \institute{
  Kahlert School of Computing\\
  University of Utah, Salt Lake City, UT 84112, USA\\
  \email{u0866264@utah.edu, haitao.wang@utah.edu}
}

\maketitle

\pagestyle{plain}
\pagenumbering{arabic}
\setcounter{page}{1}

\vspace{-0.2in}
\begin{abstract}
Given a set $P$ of $n$ points and a set $S$ of $m$ disks in the plane, the disk hitting set problem asks for a smallest subset of $P$ such that every disk of $S$ contains at least one point in the subset. The problem is NP-hard. In this paper, we consider a line-constrained version in which all disks have their centers on a line. We present an $O(m\log^2n+(n+m)\log(n+m))$ time algorithm for the problem. This improves the previously best result of $O(m^2\log m+(n+m)\log(n+m))$ time for the weighted case of the problem where every point of $P$ has a weight and the objective is to minimize the total weight of the hitting set. Our algorithm actually solves a more general line-separable problem with a single intersection property: The points of $P$ and the disk centers are separated by a line $\ell$ and the boundary of every two disks intersect at most once on the side of $\ell$ containing $P$. 
\end{abstract}

\section{Introduction}
\label{sec:intro}

Let $P$ be a set of $n$ points and $S$ a set of $m$ disks in the plane. The {\em hitting set} problem is to compute a smallest subset of $P$ such that every disk in $S$ contains at least one point in the subset (i.e., every disk is {\em hit} by a point in the subset, and the subset is called a {\em hitting set}). The problem is NP-hard, even if all disks have the same radius~\cite{ref:KarpRe72,ref:MustafaIm10}. Polynomial-time approximation algorithms are known for the problem, e.g.,~\cite{ref:OualiA14,ref:MustafaIm10,ref:BusPr18,ref:EvenHi05,ref:GanjugunteGe11,ref:LiA15}. 


In this paper, we consider the {\em line-constrained} version of the problem, where centers of all disks 
are on a line while the points of $P$ can be anywhere in the plane. 
The weighted case of the problem was studied by Liu and Wang~\cite{ref:LiuGe23}, where 
each point of $P$ has a weight and the objective is to minimize the total weight of the hitting set. Their algorithm runs in $O((m+n)\log(m+n) + \kappa\log m)$ time, where $\kappa$ is the number of pairs of disks that intersect and $\kappa=O(m^2)$ in the worst case. They reduced the runtime to $O((m+n)\log(m+n))$ for the {\em unit-disk case}, where all disks have the same radius~\cite{ref:LiuGe23}.
Our problem in this paper is for the unweighted case. To the best of our knowledge, we are not aware of any previous work that particularly studied the unweighted hitting set problem for line-constrained disks. We propose an algorithm of $O(m\log^2n+(n+m)\log(n+m))$ time, which improves the weighted case algorithm of $O(m^2\log m+(n+m)\log(n+m))$ worst-case time~\cite{ref:LiuGe23}. Perhaps theoretically more interesting is that the worst-case runtime of our algorithm is near linear. 

\subsection{Related work}

A closely related problem is the disk coverage problem, which is to compute a smallest subset of $S$ that together cover all the points of $P$. This problem is also NP-hard because it is dual to the hitting set problem in the unit-disk case (i.e., all disks have the same radius). Polynomial-time algorithms are known for certain special cases, e.g., \cite{ref:AmbuhlCo06,ref:CalinescuSe04,ref:ChanEx14,ref:DasHo10}. In particular, the line-constrained problem (in which all disks are centered on a line) was studied by Pedersen and Wang~\cite{ref:PedersenAl22}. Their algorithm runs in $O((m+n)\log(m+n) + \kappa\log m)$ time, where $\kappa$ is the number of pairs of disks that intersect and $\kappa=O(m^2)$ in the worst case; they also solved the unit-disk case in $O((m+n)\log(m+n))$ time. 
As noted above, in the unit-disk case, the coverage and hitting set problems are dual to each other and therefore the two problems can essentially be solved by the same algorithm. However, this is not the case if the radii of the disks are different.%
\footnote{Note that \cite{ref:DurocherDu15} provides a method to reduce certain coverage problems to instances of the hitting set problem; however, the reduction algorithm, which takes more than $O(n^5)$ time, is not efficient.}

In addition, the $O((m+n)\log(m+n) + \kappa\log m)$ time algorithm of Pedersen and Wang~\cite{ref:PedersenAl22} also works for the weighted {\em line-separable unit-disk} version, where all disks have the same radius and the disk centers are separated from the points of $P$ by a line. 

All the above results are for the weighted case. The unweighted disk coverage case was also particularly studied before. Liu and Wang~\cite{ref:LiuOn23}\footnote{See the arXiv version of \cite{ref:LiuOn23}, which improves the result in the original conference paper. The algorithms follow the same idea, but the arXiv version provides more efficient implementations.} considered the line-constrained problem and gave an $O(m\log n\log m+(n+m)\log(n+m))$ time algorithm. For the line-separable unit-disk case, Amb\"uhl et al.~\cite{ref:AmbuhlCo06} derived an algorithm of $O(m^2n)$ time, which was used as a subroutine in their algorithm for the general coverage problem in the plane (without any constraints).
An improved $O(nm+ n\log n )$ time algorithm is presented in \cite{ref:ClaudeAn10}. Liu and Wang's approach~\cite{ref:LiuOn23} solves this case in $O((n+m)\log (n+m))$ time. 

If disks of $S$ are half-planes, the problem becomes the half-plane coverage problem. 
For the weighted case, Chan and Grant \cite{ref:ChanEx14} proposed an algorithm for the {\em lower-only case} where all half-planes are lower ones; their algorithm runs in $O(n^4)$ time when $m=n$. 
With the observation that a half-plane may be considered as a unit disk of infinite radius, the lower-only half-plane coverage problem is essentially a special case of the line-separable unit-disk coverage problem~\cite{ref:PedersenAl22}. Consequently, applying the algorithm of \cite{ref:PedersenAl22} can solve the weighted lower-only case in $O(n^2\log n)$ time (when $m=n$) and applying the algorithm of \cite{ref:LiuOn23} can solve the unweighted lower-only case in $O(n\log n)$ time. Wang and Xue~\cite{ref:WangAl24} derived another $O(n\log n)$ time algorithm for the unweighted lower-only case with a different approach 
and also proved an $\Omega(n\log n)$ lower bound under the algebraic decision tree model by a reduction from the set equality problem~\cite{ref:Ben-OrLo83} (note that this leads to the same lower bound for the line-separable unit-disk coverage problem). 
For the general case where both upper and lower half-planes are present, Har-Peled and Lee \cite{ref:Har-PeledWe12} solved the weighted problem in $O(n^5)$ time. Pedersen and Wang~\cite{ref:PedersenAl22} showed that the problem can be reduced to $O(n^2)$ instances of the lower-only case problem. Consequently, applying the algorithms of \cite{ref:PedersenAl22} and \cite{ref:LiuOn23} can solve the weighted and unweighted cases in $O(n^4\log n)$ and $O(n^3\log n)$ time, respectively. 
Wang and Xue~\cite{ref:WangAl24} gave a more efficient algorithm of $O(n^{4/3}\log^{5/3}n\log^{O(1)}\log n)$ time for the unweighted case. 


\subsection{Our result}
Instead of solving the line-constrained problem directly, we tackle a more general problem in which the points of $P$ and the centers of the disks of $S$ are separated by a line $\ell$ such that the boundaries of every two disks intersect at most once on the side of $\ell$ containing $P$ (see Fig.~\ref{fig:singleinter}). We refer to it as the {\em line-separable single-intersection} hitting set problem (we will explain it later why this problem is more general than the line-constrained problem).
We present an algorithm of $O(m\log^2n+(n+m)\log(n+m))$ time for the problem. To this end, we find that some points in $P$ are ``useless'' and thus can be pruned from $P$. More importantly, the remaining points have certain properties so that the problem can be reduced to the 1D hitting set problem, which can then be easily solved. The algorithm itself is relatively simple and quite elegant. However, one challenge is to show its correctness, and specifically, to prove why the ``useless'' points are indeed useless. The proof is lengthy and fairly technical, which is one of our main contributions. 

\begin{figure}[t]
\begin{minipage}[t]{\textwidth}
\begin{center}
\includegraphics[height=1.0in]{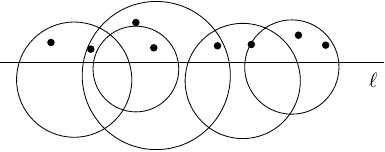}
\caption{\footnotesize Illustrating the line-separable single-intersection case: Centers of all disks are below $\ell$.}
\label{fig:singleinter}
\end{center}
\end{minipage}
\vspace{-0.15in}
\end{figure}

\paragraph{\bf The line-constrained problem.}
To solve the line-constrained problem, where all disks of $S$ are centered on a line $\ell$, the problem can be reduced to the line-separable single-intersection case. Indeed, without loss of generality, we assume that $\ell$ is the $x$-axis. 
For each point $p$ of $P$ below $\ell$, we replace $p$ by its symmetric point with respect to $\ell$. As such, we obtain a set of points that are all above $\ell$. Since all disks are centered on $\ell$, it is not difficult to see that an optimal solution using this new set of points corresponds to an optimal solution using $P$. Furthermore, since disks are centered on $\ell$, although their radii may not be equal, the boundaries of any two disks intersect at most once above $\ell$. 
Therefore, the problem becomes an instance of the line-separable single-intersection case. As such, applying the algorithm for line-separable single-intersection problem solves the line-constrained problem in $O(m\log^2n+(n+m)\log(n+m))$ time.
Therefore, in the rest of the paper, we will focus on solving the line-separable single-intersection problem.

\paragraph{\bf The unit-disk case.}
As mentioned earlier, the unit-disk case problem where all disks have the same radius can be reduced to the coverage problem (and vice versa). More specifically, if we consider the set of unit disks centered at the points of $P$ as a set of ``dual disks'' and consider the centers of the disks of $S$ a set of ``dual points'', then the hitting set problem is equivalent to finding a smallest subset of dual disks whose union covers all dual points. 
Consequently, applying the line-separable unit-disk coverage algorithm in \cite{ref:LiuOn23} solves the hitting set problem in $O((n+m)\log(n+m))$ time. Nevertheless, we show that our technique can directly solve the hitting set problem in this case in the same time complexity.

\paragraph{\bf The half-plane hitting set problem.} 
As in the coverage problem discussed above, if disks of $S$ are half-planes, the problem becomes the half-plane hitting set problem. For the weighted case, the approach of Chan and Grant \cite{ref:ChanEx14} solves the lower-only case in $O(n^4)$ time when $m=n$. Again, with the observation that a half-plane may be viewed as a unit disk of infinite radius, the lower-only half-plane hitting set problem is a special case of the line-separable unit-disk hitting set problem. As such, applying the algorithm of \cite{ref:LiuGe23} can solve the weighted lower-only case in $O(n^2\log n)$ time (when $m=n$) and applying the unit-disk case algorithm discussed above can solve the unweighted lower-only case in $O(n\log n)$ time. 
For the general case where both upper and lower half-planes are present, Har-Peled and Lee \cite{ref:Har-PeledWe12} solved the weighted problem in $O(n^6)$ time. Liu and Wang~\cite{ref:LiuGe23} showed that the problem (for both the weighted and unweighted cases) can be reduced to $O(n^2)$ instances of the lower-only case problem. Consequently, applying the above algorithms for the weighted and unweighted lower-only case problems can solve the weighted and unweighted general case problems in $O(n^4\log n)$ and $O(n^3\log n)$ time, respectively.


\paragraph{\bf Lower bound.}
As discussed above, the $\Omega(n\log n)$ lower bound in \cite{ref:WangAl24} for the lower-only half-plane coverage problem leads to the $\Omega(n\log n)$ lower bound for the line-separable unit-disk coverage problem when $m=n$. As the unit-disk hitting set problem is dual to the unit-disk coverage problem, it also has $\Omega(n\log n)$ as a lower bound. Since the unit-disk hitting set problem is a special case of the line-separable single-intersection hitting set problem, $\Omega(n\log n)$ is also a lower bound of the latter problem. Similarly, since the lower-only half-plane hitting set is dual to the lower-only half-plane coverage, $\Omega(n\log n)$ is also a lower bound of the former problem.

\paragraph{\bf An algorithm in the algebraic decision tree model.} In the algebraic decision tree model, where the time complexity is measured only by the number of comparisons, our method, combining with a technique recently developed by Chan and Zheng~\cite{ref:ChanHo23}, shows that the line-separable single-intersection problem (and thus the line-constrained problem) can be solved using $O((n+m)\log (n+m))$ comparisons, matching the above lower bound. 
To ensure clarity in the following discussion, unless otherwise stated, all time complexities are based on the standard real RAM model.


\paragraph{\bf Outline.} The rest of the paper is organized as follows. After introducing the notation in Section~\ref{sec:pre}, we describe our algorithm in Section~\ref{sec:hitting}. The correctness of the algorithm is proved in Section~\ref{sec:hitcorrect}. Finally, we show how to implement the algorithm efficiently in Section~\ref{sec:hitimplement}. The unit-disk case algorithm is discussed at the end of Section~\ref{sec:hitimplement} as a modification of our general case algorithm. The algebraic decision tree algorithm is also discussed in Section~\ref{sec:hitimplement}.

\section{Preliminaries}
\label{sec:pre}
In this section, we introduce some notation and concepts that will be used throughout the paper. 

As discussed above, we focus on the line-separable single-intersection case. Let $P$ be a set of $n$ points and $S$ a set of $m$ disks in the plane such that the points of $P$ and the centers of the disks of $S$ are separated by a line $\ell$ and the boundaries of every two disks intersect at most once on the side of $\ell$ which contains $P$. Note that the points of $P$ and the disk centers are allowed to be on $\ell$. Without loss of generality, we assume that $\ell$ is the $x$-axis and the points of $P$ are above (or on) $\ell$ while the disk centers are below (or on) $\ell$ (see Fig.~\ref{fig:singleinter}). As such, the boundaries of every two disks intersect at most once above $\ell$. Our goal is to compute a smallest subset of $P$ such that each disk of $S$ is hit by at least one point in the subset. 

Under this setting, for each disk $s\in S$, only its portion above $\ell$ matters for our problem. Hence, unless otherwise stated, a disk $s$ refers only to its portion above (and on) $\ell$. As such, the boundary of $s$ consists of an {\em upper arc}, i.e., the boundary arc of the original disk above $\ell$, and a {\em lower segment}, i.e., the intersection of $s$ with $\ell$. 
Note that $s$ has a single leftmost (resp., rightmost) point, which is the left (resp., right) endpoint of the lower segment of $s$. 

If $P'$ is a subset of $P$ that form a hitting set for $S$, we call $P'$ a {\em feasible solution}. If $P'$ is a feasible solution of minimum size, then $P'$ is an {\em optimal solution}.

We assume that each disk of $S$ is hit by at least one point of $P$ since otherwise there would be no feasible solution. Our algorithm is able to check whether the assumption is met. 

We make a general position assumption that no two points of $A$ have the same $x$-coordinate, where $A$ is the union of $P$ and the set of the leftmost and rightmost points of the upper arcs of all disks. Degenerate cases can be handled by standard perturbation techniques, e.g., \cite{ref:EdelsbrunnerSi90}.

For any point $p$ in the plane, we denote its $x$-coordinate by $x(p)$. We sort all points in $P$ in ascending order of their $x$-coordinates, resulting in a sorted list $\{p_1, p_2,\cdots, p_n\}$. We use $P[i,j]$ to denote the subset $\{p_i, p_{i+1},\cdots, p_j\}$, for any $1\leq i\leq j\leq n$. 
We sort all disks in ascending order of the $x$-coordinates of the leftmost points of their upper arcs; let $\{s_1, s_2,\cdots, s_m\}$ be the sorted list. We use $S[i,j]$ to denote the subset $\{s_i, s_{i+1},\cdots, s_j\}$, for $1\leq i\leq j\leq m$. 
For convenience, let $P[i,j]=\emptyset$ and $S[i,j]=\emptyset$ if $i>j$. 
For each disk $s_i$, let $l_i$ and $r_i$ denote the leftmost and rightmost points of its upper arc, respectively.

For any disk $s\in S$, we use $S_l(s)$ (resp., $S_r(s)$) to denote the subset of disks $S$ whose leftmost points are to the left (resp., right) of that of $s$, that is, if the index of $s$ is $i$, then $S_l(s)=S[1,i-1]$ and $S_r(s)=S[i+1,m]$. For any disk $s'\in S_l(s)$, we also say that $s'$ is {\em to the left} of $s$; similarly, if $s'\in S_r(s)$, then $s'$ is {\em to the right} of $s$. For convenience, if $s'$ is to the left of $s$, we use $s'\prec s$ to denote it. 

For a point $p_i\in P$ and a disk $s_k\in S$, we say that $p_i$ is {\em vertically above} $s_k$ (or $s_k$ is {\em vertically below} $p_i$) if $p_i$ is outside $s_k$ and $x(l_k) < x(p_i) < x(r_k)$. 

\paragraph{\bf The non-containment property.}
If a disk $s_i$ contains another disk $s_j$ completely, then $s_i$ is redundant for our problem since any point hitting $s_j$ also hits $s_i$. It is easy to find those redundant disks in $O(m\log m)$ time (indeed, this is a 1D problem since $s_i$ contains $s_j$ if and only if the lower segment of $s_i$ contains that of $s_j$). Therefore, to solve our problem, we first remove such redundant disks from $S$ and then work on the remaining disks. 
For simplicity, from now on we assume that no disk of $S$ contains another. Therefore, $S$ has the following {\em non-containment} property, which is critical to our algorithm. 

\begin{observation}{\em (Non-Containment Property)}\label{obser:FIFO}
For any two disks $s_i, s_j \in S$, $x(l_i) < x(l_j)$ if and only if $x(r_i) < x(r_j)$.
\end{observation}


\section{The algorithm description}
\label{sec:hitting}

In this section, we describe our algorithm.
We follow the notation defined in Section~\ref{sec:pre}. 

We begin with the following definition, which is critical for our algorithm. 

\begin{definition}
    For each disk $s_i\in S$, among all the points of $P$ covered by $s_i$, define $a(i)$ as the smallest index of these points and $b(i)$ the largest index of them. 
\end{definition}

Since each disk $s_i$ contains at least one point of $P$, both $a(i)$ and $b(i)$ are well defined. 

\begin{definition}\label{def:prune}
   For any point $p_k\in P$, we say that $p_k$ is {\em prunable} if there is a disk $s_i\in S$ such that $p_k\not\in s_i$ and $a(i)<k<b(i)$. 
\end{definition}

In the following, we describe our algorithm. Although the algorithm description seems simple, establishing its correctness is by no means an easy task. 
We devote Section~\ref{sec:hitcorrect} to the correctness proof. The implementation of the algorithm, which is also not straightforward, is presented in Section~\ref{sec:hitimplement}. 

\paragraph{\bf Algorithm description.} 
The algorithm has three main steps. 

\begin{enumerate}
    \item Compute $a(i)$ and $b(i)$ for all disks $s_i\in S$. We will show in Section~\ref{sec:hitimplement} that this can be done in $O(m\log^2n+(n+m)\log(n+m))$ time. 

    \item Find all prunable points; let $Q$ be the set of all prunable points. We will show in Section~\ref{sec:hitimplement} that $Q$ can be computed in $O((n+m)\log(n+m))$ time. 
    
    Let $P^*=P\setminus Q$. We will prove in Section~\ref{sec:hitcorrect} that $P^*$ contains an optimal solution to the hitting problem on $P$ and $S$. This means that it suffices to work on $P^*$ and $S$.
    
    \item Reduce the hitting set problem on $P^*$ and $S$ to a 1D hitting set problem, as follows. For each point of $P^*$, we project it perpendicularly onto $\ell$. Let $\tP$ be the set of all projected points. For each disk $s_i\in S$, we create a segment on $\ell$ whose left endpoint has $x$-coordinate equal to $x(p_{a(i)})$ and whose right endpoint has $x$-coordinate equal to $x(p_{b(i)})$. Let $\tS$ be the set of all segments thus created. 

    We solve the following 1D hitting set problem: Find a smallest subset of points of $\tP$ such that every segment of $\tS$ is hit by a point of the subset. This 1D problem can be easily solved in $O((|\tS|+|\tP|)\log (|\tS|+|\tP|))$ time~\cite{ref:LiuGe23},\footnote{The algorithm in \cite{ref:LiuGe23}, which uses dynamic programming, is for the weighted case where each point has a weight. Our problem is simpler because it is the unweighted case. We can use a simple greedy algorithm to solve it.} which is $O((m+n)\log(m+n))$ since $|\tP|\leq n$ and $|\tS|= m$. 

    Suppose that $\tP_{\text{opt}}$ is any optimal solution for the 1D problem. We create a subset $P^*_{\text{opt}}$ of $P^*$ as follows. For each point of $\tP_{\text{opt}}$, suppose that it is the projection of a point $p_i\in P^*$; then we add $p_i$ to $P^*_{\text{opt}}$. We will prove in Section~\ref{sec:hitcorrect} that $P^*_{\text{opt}}$ is an optimal solution to the hitting set problem on $P^*$ and $S$.
\end{enumerate}

We summarize the result in the following theorem. 

\begin{theorem}\label{theo:hit}
Given a set $P$ of $n$ points and a set $S$ of $m$ disks in the plane such that the disk centers are separated from the points of $P$ by a line and the single-intersection condition is satisfied, 
the hitting set problem for $P$ and $S$ is solvable in $O(m\log^2n+(n+m)\log(n+m))$ time.
\end{theorem}

\section{Algorithm correctness}
\label{sec:hitcorrect}

In this section, we prove the correctness of our algorithm. More specifically, we will argue the correctness of the second and the third main steps of the algorithm. 

In the following, we start with the third main step, as it is relatively straightforward. In fact, arguing the correctness of the second main step is quite challenging and is a main contribution of our paper. 

\paragraph{\bf Correctness of the third main step.}

For each disk $s_i\in S$, let $s_i'$ refer to the segment of $\tS$ created from $s_i$. For each point $p_j\in P$, let $p_j'$ refer to the point of $\tP$ which is the projection of $p_j$. The following lemma justifies the correctness of the third main step of the algorithm. 

\begin{lemma}
    A point $p_j\in P^*$ hits a disk $s_i\in S$ if and only if $p_j'$ hits $s_i'$. 
\end{lemma}
\begin{proof}
    Suppose $p_j$ hits $s_i$. Then, $p_j\in s_i$. By the definitions of $a(i)$ and $b(i)$, we have $a(i)\leq j\leq b(i)$. Hence, $x(p_{a(i)})\leq x(p_j)\leq x(p_{b(i)})$, and thus $p_j'$ hits $s_i'$ by the definitions of $p_j'$ and $s_i'$. 

    On the other hand, suppose that $p_j'$ hits $s_i'$. Then, according to the definitions of $p_j'$ and $s_i'$, $x(p_{a(i)})\leq x(p_j)\leq x(p_{b(i)})$ holds. If $j=a(i)$ or $j=b(i)$, then $p_j$ must hit $s_i$ following the definitions of $a(i)$ and $b(i)$. Otherwise, we have $a(i)< j< b(i)$. 
    Observe that $p_j$ must be inside $s_i$ since otherwise $p_j$ would be a prunable point and therefore could not be in $P^*$. As such, $p_j$ must hit $s_i$.
    \qed
\end{proof}

\subsection{Correctness of the second main step}

In what follows, we focus on the correctness of the second main step. 

For any disk $s$, let $P(s)$ denote the subset of points of $P$ inside $s$. For any point $p$, let $S(p)$ denote the subset of disks of $S$ hit by $p$. For any subset $P'\subseteq P$, by slightly abusing notation, let $S(P')$ denote the subset of disks of $S$ hit by at least one point of $P'$, i.e., $S(P')=\bigcup_{p\in P'}S(p)$. 

The following observation follows directly from the definition of prunable points. 

\begin{observation}\label{obser:hitprune}
    Suppose a point $p$ is a prunable point in $P$. Then, there is a disk $s\in S$ vertically below $p$ such that the following are true.  
    \begin{enumerate}    
        \item $P(s)$ has both a point left of $p$ and a point right of $p$.
        \item 
        For any two points $p^l, p^r\in P(s)$ with one left of $p$ and the other right of $p$, we have $S(p)\subseteq S(p^l)\cup S(p^r)$.         
    \end{enumerate}
\end{observation}
\begin{proof}
The first statement directly follows the definition of prunable points. For the second statement, without loss of generality, assume that $p^l$ is left of $p$ while $p^r$ is right of $p$ (see Fig.~\ref{fig:obser2}). Consider any disk $s_i\in S(p)$. By definition, $p\in s_i$. 
As $p\not\in s$, $s_i\neq s$. Hence, $s_i$ is either in $S_l(s)$ or in $S_r(s)$. If $s_i\in S_l(s)$, then due to the non-containment property, $s_i$ must contain the area of $s$ to the left of $p$ and therefore must contain $p^l$, which implies $s_i\in S(p^l)$. Similarly, if $s_i\in S_r(s)$, then $s_i$ must be in $S(p^r)$. 
\qed
\end{proof}

\begin{figure}[t]
\begin{minipage}[t]{\textwidth}
\begin{center}
\includegraphics[height=0.9in]{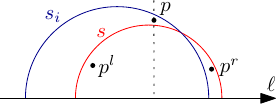}
\caption{\footnotesize Illustrating the proof of Observation~\ref{obser:hitprune}.}
\label{fig:obser2}
\end{center}
\end{minipage}
\vspace{-0.15in}
\end{figure}

The following lemma establishes the correctness of the second main step of the algorithm. 

\begin{lemma}\label{lem:hitprune}
$P^*$ contains an optimal solution for the hitting set problem on $S$ and $P$. 
\end{lemma}
\begin{proof}
Let $\optp$ be an optimal solution for $S$ and $P$. Let $Q$ be the set of all prunable points. Recall that $P^*=P\setminus Q$. If $\optp\cap Q=\emptyset$, then $\optp\subseteq P^*$ and thus the lemma is vacuously true. In what follows, we assume that $|\optp\cap Q|\geq 1$. 

Pick an arbitrary point from $\optp\cap Q$, denoted by $\hp_1$. Below, we give a process that can find a point $p^*$ from $P^*$ to replace $\hp_1$ in $\optp$ such that the new set $\optp^1=\{p^*\}\cup \optp\setminus\{\hp_1\}$ is a feasible solution, implying that $\optp^1$ is still an optimal solution since $|\optp^1|=|\optp|$. As $p^*\in P^*$, we have $|\optp^1\cap Q|=|\optp\cap Q|-1$. Therefore, if $\optp^1\cap Q$ is still nonempty, then we can repeat the process for other points in $\optp^1\cap Q$ until we obtain an optimal solution $\optp^*$ with $\optp^*\cap Q=\emptyset$, which will prove the lemma. 
The process involves induction. To help the reader understand it better, we first provide the details for the first two iterations of the process (we will introduce some notation that appears unnecessary for the first two iterations, but these will be needed for explaining the inductive hypothesis later). 

\paragraph{\bf The first iteration.}
Let $\optp'=\optp\setminus\{\hp_1\}$. Since $\hp_1\in Q$, by Observation~\ref{obser:hitprune}, $S$ has a disk $\hs_1$ vertically below $\hp_1$ such that $P(\hs_1)$ contains both a point left of $\hp_1$, denoted by $\hp_1^l$, and a point right of $\hp_1$, denoted by $\hp_1^r$. Furthermore, $S(\hp_1)\subseteq S(\hp_1^l)\cup S(\hp_1^r)$. Since $\hp_1\not\in \hs_1$ and $\optp=\optp'\cup \{\hp_1\}$ forms a hitting set of $P$, $\optp'$ must have a point $p$ that hits $\hs_1$. Clearly, $p$ is left or right of $\hp_1$. Without loss of generality, we assume that $p$ is right of $\hp_1$. Since $\hp_1^r$ refers to an arbitrary point to the right of $\hp_1$ that hits $\hs_1$ and $p$ is also a point to right of $\hp_1$ that hits $\hs_1$, for notational convenience, we let $\hp_1^r$ refer to $p$. As such, $\hp_1^r$ is in $\optp'$. 

Consider the point $\hp_1^l$. Since $S(\hp_1)\subseteq S(\hp_1^l)\cup S(\hp_1^r)$ and $\hp_1^r$ is in $\optp'$, it is not difficult to see that $S(\optp)\subseteq S(\optp')\cup S(\hp_1^l)$ and thus $\optp'\cup \{\hp_1^l\}$ is a feasible solution. As such, if $\hp_1^l\not\in Q$, then we can use $\hp_1^l$ as our target point $p^*$ and our process (to find $p^*$) is performed. In what follows, we assume $\hp_1^l\in Q$. 

We let $\hp_2=\hp_1^l$. Define $A_1=\{\hp_1^r\}$. According to the above discussion, $A_1\subseteq \optp'$, $S(\hp_1)\subseteq S(A_1)\cup S(\hp_2)$, $\optp'\cup \{\hp_2\}$ is a feasible solution, $\hp_1$ is vertically above $\hs_1$, and $\hp_2\in \hs_1$.  

\paragraph{\bf The second iteration.}
We are now entering the second iteration of our process. First, notice that $\hp_2$ cannot be $\hp_1$ since $\hp_2=\hp_1^l$, which cannot be $\hp_1$. Our goal in this iteration is to find a {\em candidate point} $p'$ to replace $\hp_2$ so that $\optp'\cup \{p'\}$ also forms a hitting set of $S$. Consequently, if $p'\not\in Q$, then we can use $p'$ as our target $p^*$; otherwise, we need to guarantee $p'\neq \hp_1$ so that our process will not enter a loop. 
The discussion here is more involved than in the first iteration. 

Since $\hp_2\in Q$, by Observation~\ref{obser:hitprune}, $S$ has a disk $\hs_2$ vertically below $\hp_2$ 
such that $P(\hs_2)$ contains both a point left of $\hp_2$, denoted by $\hp_2^l$, and a point right of $\hp_2$, denoted by $\hp_2^r$. Furthermore, $S(\hp_2)\subseteq S(\hp_2^l)\cup S(\hp_2^r)$. 
Depending on whether $\hs_2$ is in $S(A_1)$, there are two cases. 

\begin{itemize}
    \item 
If $\hs_2\not\in S(A_1)$, since $\hp_2$ does not hit $\hs_2$ and $S(\hp_1)\subseteq S(A_1)\cup S(\hp_2)$, we obtain $\hs_2\not\in S(\hp_1)$. Now we can basically repeat our argument from the first iteration. 
Since $\hp_2$ does not hit $\hs_2$ and $\optp'\cup \{\hp_2\}$ is a feasible solution, $\optp'$ must have a point $p$ that hits $\hs_2$. Clearly, $p$ is either left or right of $\hp_2$. 

We first assume that $p$ is right of $\hp_2$. Since $\hp_2^r$ refers to an arbitrary point to the right of $\hp_2$ that hits $\hs_2$ and $p$ is also a point to right of $\hp_2$ that hits $\hs_2$, for notational convenience, we let $\hp_2^r$ refer to $p$. As such, $\hp_2^r$ is in $\optp'$. 

We let $\hp_2^l$ be our candidate point, which satisfies our need as discussed above for $p'$. Indeed, since $\optp'\cup \{\hp_2\}$ is an optimal solution, $S(\hp_2)\subseteq S(\hp_2^l)\cup S(\hp_2^r)$, and $\hp_2^r\in \optp'$, we obtain that $\optp'\cup \{\hp_2^l\}$ also forms a hitting set of $S$. Furthermore, since $\hp_2^l$ hits $\hs_2$ while $\hp_1$ does not, we know that $\hp_2^l\neq \hp_1$. Therefore, if $\hp_2^l\not\in Q$, then we can use $\hp_2^l$ as our target $p^*$ and we are done with the process. Otherwise, we let $\hp_3=\hp_2^l$ and then enter the third iteration. In this case, we let $A_2=A_1\cup \{\hp_2^r\}$. According to our above discussion, $A_2\subseteq \optp'$, $S(\hp_2)\subseteq S(A_2)\cup S(\hp_3)$, $\{\hp_3\}\cup \optp'$ is a feasible solution, $\hp_2$ is vertically above $\hs_2$, and $\hp_3\in \hs_2$.  

The above discussed the case where $p$ is right of $\hp_2$. 
If $p$ is left of $\hp_2$, then the analysis is symmetric.\footnote{More specifically, if $\hp_2^r \not\in Q$, then we can use $\hp_2^r$ as our target $p^*$ and the process is complete. Otherwise, we let $\hp_3 = \hp_2^r$ and enter the third iteration; in this case, we let $A_2=A_1\cup \{\hp_2^l\}$.}

\item 
If $\hs_2\in S(A_1)$, then we let $\hp_2^l$ be our candidate point. We show in the following that it satisfies our need as discussed above for $p'$, i.e., $\{\hp_2^l\}\cup \optp'$ forms a hitting set of $S$ and $\hp_2^l\neq \hp_1$. 

Indeed, since $A_1=\{\hp_1^r\}$ and $\hs_2\in S(A_1)$, $\hs_2$ is hit by $\hp_1^r$. Since $\hp_1^r$ is to the right of $\hp_1$, and $\hp_2$, which is $\hp_1^l$, is to the left of $\hp_1$, we obtain that $\hp_1^r$ is to the right of $\hp_2$. 
Since $\hp_2^l$ hits $\hs_2$, $\hp_2^l$ is to the left of $\hp_2$, $\hp_1^r$ hits $\hs_2$, and $\hp_1^r$ is to the right of $\hp_2$, by Observation~\ref{obser:hitprune}, $S(\hp_2)\subseteq S(\hp_2^l) \cup S(\hp_1^r)$, i.e., $S(\hp_2)\subseteq S(\hp_2^l) \cup S(A_1)$. Since $\optp'\cup \{\hp_2\}$ is a feasible solution and $A_1\subseteq \optp'$, it follows that $\{\hp_2^l\}\cup \optp'$ is also a feasible solution. On the other hand, since $\hp_2^l$ is to the left of $\hp_2$ while $\hp_2$ (which is $\hp_1^l$) is to the left of $\hp_1$, we know that $\hp_2^l$ is to the left of $\hp_1$ and thus 
$\hp_2^l\neq \hp_1$. 

As such, if $\hp_2^l\not\in Q$, we can use $\hp_2^l$ as our target $p^*$ and we are done with the process. Otherwise, we let $\hp_3=\hp_2^l$ and continue with the third iteration. In this case, we let $A_2=A_1$. According to our above discussion, $A_2\subseteq \optp'$, $S(\hp_2)\subseteq S(A_2)\cup S(\hp_3)$, $\optp'\cup \{\hp_3\}$ is a feasible solution, $\hp_2$ is vertically above $\hs_2$, and $\hp_3\in \hs_2$.  
\end{itemize}

This finishes the second iteration of the process.

\paragraph{\bf Inductive step.} In general, suppose that we are entering the $i$-th iteration of the process with the point $\hp_i\in Q$, $i\geq 2$. We make the following inductive hypothesis for $i$.

\begin{enumerate}
    \item We have points $\hp_k\in Q$ for all $k=1,2,\ldots,i-1$ in the previous $i-1$ iterations such that $\hp_i\neq \hp_k$ for any $1\leq k\leq i-1$
    \item We have subsets $A_k$ for all $k=1,2,\ldots,i-1$ such that $A_1\subseteq A_2\subseteq\cdots \subseteq A_{i-1}\subseteq\optp'$, and $S(\hp_k)\subseteq S(A_k)\cup S(\hp_{k+1})$ holds for each $1\leq k\leq i-1$. 
    \item For any $1\leq k\leq i$, $\{\hp_k\}\cup \optp'$ is a feasible solution.  
    \item We have disks $\hs_k\in S$ for $k=1,2,\ldots,i-1$ such that $\hs_k$ is vertically below $\hp_k$ and $\hp_{k+1}\in \hs_k$.
    \end{enumerate}

Our previous discussion already established the hypothesis for $i=2$ and $i=3$. 
Next, we proceed with the $i$-th iteration argument for any general $i\geq 4$. Our goal is to find a candidate point $\hp_{i+1}$ such that $\optp'\cup\{\hp_{i+1}\}$ is a feasible solution and the inductive hypothesis still holds for $i+1$. 

Since $\hp_i\in Q$, by Observation~\ref{obser:hitprune}, there is a disk $\hs_i$ vertically below $\hp_i$ such that $P(\hs_i)$ has a point left of $\hp_i$, denoted by $\hp_i^l$, and a point right of $\hp_i$, denoted by $\hp_i^r$.
Furthermore, $S(\hp_i)\subseteq S(\hp_i^l)\cup S(\hp_i^r)$. 
Depending on whether $\hs_i$ is in $S(A_{i-1})$, there are two cases. 

\begin{enumerate}
    \item If $\hs_i\not\in S(A_{i-1})$, then since $\hs_i$ does not contain $\hp_i$ and $\optp'\cup \{\hp_i\}$ is a feasible solution, $\optp'$ must have a point $p$ that hits $\hs_i$. Clearly, $p$ is to the left or right of $\hp_i$. Without loss of generality, we assume that $p$ is to the right of $\hp_i$. Since $\hp_i^r$ refers to an arbitrary point to the right of $\hp_i$ that hits $\hs_i$ and $p$ is also a point to the right of $\hp_i$ that hits $\hp_i$, for notational convenience, we let $\hp_i^r$ refer to $p$. As such, $\hp_i^r$ is in $\optp'$. 

    We let $\hp_{i+1}$ be $\hp_{i}^l$ and define $A_{i}=A_{i-1}\cup \{\hp_{i}^r\}$. In the following, we argue that the inductive hypothesis holds. 

    \begin{itemize}
    \item First of all, by definition, $\hp_i$ is vertically above $\hs_i$ and $\hp_{i+1}\in \hs_i$. Hence, the fourth statement of the hypothesis holds. 
        \item 
        Since $\{\hp_i\}\cup \optp'$ is a feasible solution, $S(\hp_i)\subseteq S(\hp_i^l)\cup S(\hp_i^r)$, $\hp_i^r\in \optp'$, and $\hp_{i+1}=\hp_{i}^l$, we obtain that $\{\hp_{i+1}\}\cup \optp'$ is a feasible solution. This proves the third statement of the hypothesis. 
        \item
        Since $A_{i}=A_{i-1}\cup \{\hp_{i}^r\}$, $A_{i-1}\subseteq \optp'$ by inductive hypothesis, and $\hp_i^r\in \optp'$, we obtain $A_i\subseteq \optp'$. Furthermore, since $S(\hp_i)\subseteq S(\hp_i^l)\cup S(\hp_i^r)$, $\hp_i^r\in A_{i}$, and $\hp_{i+1}=\hp_{i}^l$, we have $S(\hp_i)\subseteq S(A_i)\cup S(\hp_{i+1})$.
        This proves the second statement of the hypothesis. 
        \item 
        For any point $\hp_k$ with $1\leq k\leq i-1$, to prove the first statement of the hypothesis, we need to show that $\hp_k\neq \hp_{i+1}$. To this end, since $\hs_i$ is hit by $\hp_{i+1}$, it suffices to show that $\hs_i$ is not hit by $\hp_k$. Indeed, by the inductive hypothesis, $S(\hp_k)\subseteq S(A_k)\cup S(\hp_{k+1})$ and $S(\hp_{k+1}) \subseteq S(A_{k+1})\cup S(\hp_{k+2})$. Hence, $S(\hp_k)\subseteq S(A_k)\cup S(A_{k+1})\cup S(\hp_{k+2})$. As $S(A_k)\subseteq S(A_{k+1})$, we obtain $S(\hp_k)\subseteq S(A_{k+1})\cup S(\hp_{k+2})$. Following the same argument, we can derive $S(\hp_k)\subseteq S(A_{i-1})\cup S(\hp_{i})$. Now that $\hs_i\not\in S(A_{i-1})$ and $\hs_i$ is not hit by $\hp_i$, we obtain that $\hs_i$ is not hit by $\hp_k$. 
\end{itemize}

\item 
If $\hs_i\in S(A_{i-1})$, then $\hs_i$ is hit by a point of $A_{i-1}$, say $p$. As $\hp_i\not\in\hs_i$, $p$ is left or right of $\hp_i$. Without loss of generality, we assume that $p$ is to the right of $\hp_i$. 

We let $\hp_{i+1}$ be $\hp_{i}^l$ and define $A_{i}=A_{i-1}$. We show in the following that the inductive hypothesis holds. 
\begin{itemize}
\item By definition, $\hp_i$ is vertically above $\hs_i$ and $\hp_{i+1}\in \hs_i$. Hence, the fourth statement of the hypothesis holds. 
    \item 
    Since $\hs_i$ is hit by both $p$ and $\hp_i^l$, $\hp_i^l$ is to the left of $\hp_i$, and $p$ is to the right of $\hp_i$, by Observation~\ref{obser:hitprune}, $S(\hp_i)\subseteq S(\hp_i^l)\cup S(p)$.     
    Further, since $\{\hp_i\}\cup \optp'$ is a feasible solution, $p\in A_{i-1}\subseteq \optp'$, and $\hp_{i+1}=\hp_{i}^l$, we obtain that $\{\hp_{i+1}\}\cup \optp'$ is also a feasible solution. This proves the third statement of the hypothesis.  

    \item 
    Since $A_{i-1}\subseteq \optp'$ by the inductive hypothesis and $A_i=A_{i-1}$, we have $A_i\subseteq \optp'$. As discussed above, $S(\hp_i)\subseteq S(\hp_i^l)\cup S(p)$. Since $p\in A_{i-1}=A_i$ and $\hp_{i+1}=\hp_{i}^l$, we obtain $S(\hp_i)\subseteq S(A_i)\cup S(\hp_{i+1})$. This proves the second statement of the hypothesis.  

    \item 
    For any point $\hp_k$ with $1\leq k\leq i-1$, to prove the first statement of the hypothesis, we need to show that $\hp_k\neq \hp_{i+1}$. Depending on whether $x(\hp_i) < x(\hp_k)$, there are two cases (note that $\hp_i\neq \hp_k$ by our hypothesis and thus $x(\hp_i)\neq x(\hp_k)$ due to our general position assumption). 
     \begin{itemize}
         \item If $x(\hp_i) < x(\hp_k)$, then since $\hp_{i+1}=\hp_{i}^l$, we have $x(\hp_{i+1}) < x(\hp_i)<x(\hp_k)$. Hence, $\hp_k\neq \hp_{i+1}$.

         \item If $x(\hp_k) < x(\hp_i)$, then we will show that $\hp_k \notin \hs_i$. This implies that $\hp_k \neq \hp_{i+1}$ as $\hp_{i+1} \in \hs_i$. Since the proof of $\hp_k \notin \hs_i$ is quite technical and lengthy, we devote Section~\ref{sec:LemObserProof} to it. Indeed, this is the most challenging part of the paper. 
     \end{itemize}

     This proves the first statement of the hypothesis.  
\end{itemize}
\end{enumerate}

This proves that the inductive hypothesis still holds for $i+1$. 

\bigskip

According to the inductive hypothesis, each iteration of the process finds a new candidate point $\hp_i$ such that $\optp'\cup \{\hp_i\}$ is a feasible solution. If $\hp_i\not\in Q$, then we can use $\hp_i$ as our target point $p^*$ and we are done with the process. Otherwise, we continue with the next iteration. Since each iteration finds a new candidate point (that was never used before) and $|Q|$ is finite, eventually we will find a candidate point $\hp_i$ that is not in $Q$. 

This completes the proof of the lemma. 
\qed
\end{proof}

\subsection{Proving $\hp_k \notin \hs_i$ when $x(\hp_k) < x(\hp_i)$}
\label{sec:LemObserProof}

In this section, we prove $\hp_k \notin \hs_i$ under the condition that $x(\hp_k) < x(\hp_i)$, where $1\leq k\leq i-1$. In what follows, we assume that $x(\hp_k) < x(\hp_i)$ holds. 
We start with several observations that will be useful later. 

The following observation follows the non-containment property of $S$ (recall that for two disks $s,s'\in S$, $s\prec s'$ means that $s$ is to the left of $s'$). 



   
    

\begin{observation}\label{obser:fifo10}
    Suppose $s$ and $s'$ are two disks of $S$ and there is a point $p$ that is in $s'$ and vertically above $s$. Then we have the following (see Fig.~\ref{fig:prunePoint120}):

    \begin{enumerate}
        \item Suppose $s\prec s'$. Then, for any point $p'$ with $x(p)<x(p')$, if $p'\not\in s'$, then $p'\not\in s$; if $p'\in s$, then $p'\in s'$. 

        \item Suppose $s'\prec s$. Then for any point $p'$ with $x(p')<x(p)$, if $p'\not\in s'$, then $p'\not\in s$; if $p'\in s$, then $p'\in s'$. 
    \end{enumerate}
\end{observation}
\begin{proof}
We only prove the case $s\prec s'$ as the proof of the other case is analogous. 

Since $p$ is in $s'$ and vertically above $s$, due to the non-containment property of $S$, the upper arcs of $s$ and $s'$ must intersect at a point, denoted by $q$ (see Fig.~\ref{fig:prunePoint120} left), and further, the area of $s$ to the right of $q$ must be contained inside $s'$. 
As $s\prec s'$ and $p$ is in $s'$ but not $s$, it follows that $x(q)<x(p)$. 

Consider a point $p'$ with $x(p)<x(p')$. As $x(q)<x(p)$, we have $x(q)<x(p')$. Since the area of $s$ to the right of $q$ must be contained within $s'$, if $p'\in s$, then $p'\in s'$; conversely, if $p'\not\in s'$, then $p'\not\in s$. 
\qed
\end{proof}


\begin{figure}[t]
\begin{minipage}[t]{\textwidth}
\begin{center}
\includegraphics[height=1.0in]{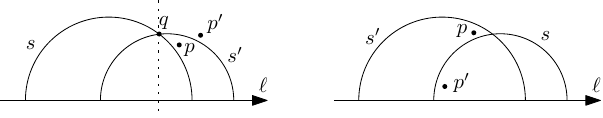}
\caption{\footnotesize Illustrating the Observation~\ref{obser:fifo10}. Left: $s\prec s'$. Right: $s'\prec s$.}
\label{fig:prunePoint120}
\end{center}
\end{minipage}
\vspace{-0.15in}
\end{figure}

\begin{observation}\label{obser:diskSstar}
    There exists a disk $s^*\in S$ that is hit by all points $\hp_j$ with $1\leq j\leq i+1$, but $s^* \notin S(A_i)$.
\end{observation}

\begin{proof}
Since $A_i=A_{i-1}$ and $S(\hp_i)\subseteq S(A_i)\cup S(\hp_{i+1})$, according to the inductive hypothesis, we have $S(\hp_1) \subseteq S(A_1) \cup S(\hp_2) \subseteq S(A_2) \cup S(\hp_3) \subseteq \cdots \subseteq S(A_i) \cup S(\hp_{i+1})$. As $S(\hp_1) \subseteq S(A_i) \cup S(\hp_{i+1})$, we claim that there must exist a disk $s^*\in S$ that is hit by both $\hp_1$ and $\hp_{i+1}$, but $s^*\not\in S(A_i)$. Indeed, assume to the contrary that this is not true. Then, $S(\hp_1) \subseteq S(A_i)$. Recall that $\{\hp_1\} \cup \optp'$ is a feasible solution and $A_i\subseteq \optp'$. Since $S(\hp_1) \subseteq S(A_i)$, we obtain that $\optp'$ is also a feasible solution. But this incurs contradiction since $\optp$ is an optimal solution and $|\optp'|=|\optp|-1$.

  Now consider any $\hp_j$ with $2 \le j \le i$. Since $S(\hp_1) \subseteq S(A_{j-1}) \cup S(\hp_j)$, $A_{j-1} \subseteq A_i$, $s^*\in S(\hp_1)$, and $s^*\not\in S(A_i)$, we obtain $s^*\in S(\hp_j)$ and thus $\hp_j$ hits $s^*$. The observation thus follows. 
    \qed
\end{proof}

\begin{observation}\label{obser:diskS_star}
      For any $j$ with $1\leq j \leq i$, if $x(\hp_{j+1})<x(\hp_j)$, then $s^*\prec\hs_j$ (see Fig.~\ref{fig:obser4}); otherwise, $\hs_j\prec s^*$.
\end{observation}

\begin{figure}[h]
\begin{minipage}[t]{\textwidth}
\begin{center}
\includegraphics[height=0.8in]{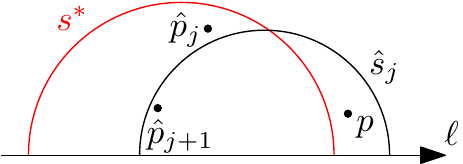}
\caption{\footnotesize Illustrating the Observation~\ref{obser:diskS_star}: If $x(\hp_{j+1}) < x(\hp_j)$, then $s^* \prec \hs_j$.}
\label{fig:obser4}
\end{center}
\end{minipage}
\vspace{-0.15in}
\end{figure}

\begin{proof}
    If $x(\hp_{j+1}) < x(\hp_j)$, then $\hp_{j+1}$ was served as $\hp^l_j$ in our process and there is a point $p \in A_j$ to the right of $\hp_j$ with $p\in \hs_j$ (to serve as $\hp^r_j$; see Fig.~\ref{fig:obser4}). Since $A_j \subseteq A_i$ and $s^*\not\in S(A_i)$ by Observation~\ref{obser:diskSstar}, we have $s^*\not\in S(A_j)$. As $p\in A_j$, $p$ does not hit $s^*$ by Observation~\ref{obser:diskSstar}. Also by Observation~\ref{obser:diskSstar}, $\hp_{j+1}$ hits $s^*$. Since $x(\hp_{j+1})<x(p)$, $\hp_{j+1}\in \hs_j$, $p\in \hs_j$, $\hp_{j+1}\in s^*$, and $p\not\in s^*$, due to the non-containment property, it must hold that $s^*\prec \hs_j$. 

    If $x(\hp_{j}) < x(\hp_{j+1})$, by following a symmetric analysis we can obtain that $\hs_j\prec s^*$. 
    \qed
\end{proof}



\begin{observation}\label{obser:si}
$s^*\prec \hs_i$ holds.
\end{observation}
\begin{proof}
Recall that $\hp_{i+1}=\hp_i^l$ and thus $x(\hp_{i+1})<x(\hp_i)$. By Observation~\ref{obser:diskS_star}, we have $s^*\prec \hs_i$.\qed
\end{proof}

With the above observations, we proceed to prove that $\hp_k \notin \hs_i$, where $1\leq k\leq i-1$. Recall that $x(\hp_k) < x(\hp_i)$. We will first prove the case $k=i-1$ and then prove the case $1\leq k\leq i-2$, because the proof of the latter case replies on the first case. 

\paragraph{\bf The case $\boldsymbol{k=i-1}$.} Since $x(\hp_{i-1}) < x(\hp_i)$, by Observation~\ref{obser:diskS_star}, $\hs_{i-1}\prec s^*$. Further, as $s^*\prec \hs_i$ by Observation~\ref{obser:si}, we obtain $\hs_{i-1}\prec \hs_i$.
Our goal is to prove $\hp_{i-1} \notin \hs_i$. Assume to the contrary that $\hp_{i-1} \in \hs_i$ (see Fig.~\ref{fig:prunePoint10}). Then, since $\hp_{i-1}$ is in $\hs_i$ and vertically above $\hs_{i-1}$, $\hs_{i-1}\prec \hs_i$, 
$x(\hp_{i-1}) < x(\hp_i)$, and $\hp_{i}\not\in \hs_i$, we obtain $\hp_{i} \notin \hs_{i-1}$ by Observation~\ref{obser:fifo10}. This incurs contradiction since $\hp_{i} \in \hs_{i-1}$ by the definition of $\hp_i$ (i.e., the inductive hypothesis in the proof of Lemma~\ref{lem:hitprune}). 

\begin{figure}[h]
\begin{minipage}[t]{\textwidth}
\begin{center}
\includegraphics[height=0.9in]{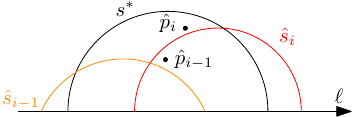}
\caption{\footnotesize Illustrating the case $k=i-1$.}
\label{fig:prunePoint10}
\end{center}
\end{minipage}
\vspace{-0.15in}
\end{figure}

\paragraph{\bf The case $\boldsymbol{1\leq k\leq i-2}$.} Our goal is to prove $\hp_k \notin \hs_i$. Recall that $x(\hp_k) < x(\hp_i)$. Assume to the contrary that $\hp_k \in \hs_i$. We have the following lemma.

    \begin{lemma}
    \label{lemma:lmPruningPoints}
    Suppose $\hp_{k}\in \hs_i$ and $x(\hp_{k}) < x(\hp_i)$ with $1\leq k\leq i-2$; then we have the following for any $j \in [k+1,i-1]$ (see Fig.~\ref{fig:prunePoint2}):
    \begin{enumerate}
            \item  If $x(\hp_{k})< x(\hp_j)$, then $\hp_j \in \hs_i$ and there exists a disk $s^1_j\in S$ such that $s^1_j\prec s^*$, $\hp_j \in s^1_j$, $\hp_{k} \notin s^1_j$, and $\hp_i \notin s^1_j$.
    
            \item  If $x(\hp_j) < x(\hp_{k})$, then there exists a disk $s^2_j\in S$ such that $s^*\prec s^2_j$, $\hp_j \in s^2_j$, and $\hp_{k} \notin s^2_j$; further, $\hp_i \notin s^2_j$ if $s^2_j\prec \hs_i$. 
    \end{enumerate}

\end{lemma}
     \begin{figure}[h]
    \begin{minipage}[t]{\textwidth}
    \begin{center}
    \includegraphics[height=0.8in]{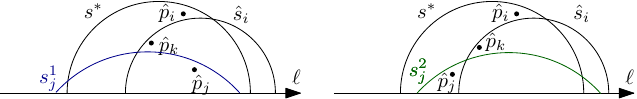}
    \caption{\footnotesize Illustrating Lemma~\ref{lemma:lmPruningPoints}. Left: $x(\hp_{k})< x(\hp_j)$. Right: $x(\hp_j) < x(\hp_{k})$.}
    \label{fig:prunePoint2}
    \end{center}
    \end{minipage}
    \vspace{-0.15in}
    \end{figure}


Before proving Lemma~\ref{lemma:lmPruningPoints}, we first use it to obtain a contradiction (and thus our assumption $\hp_k \in \hs_i$ is false). Depending on whether $x(\hp_k) < x(\hp_{i-1})$ or $x(\hp_{i-1}) < x(\hp_k)$, there are two cases. 
\begin{enumerate}
    \item If $x(\hp_k) < x(\hp_{i-1})$, then 
     by Lemma~\ref{lemma:lmPruningPoints} (with $j=i-1$), $\hp_{i-1} \in \hs_i$ and there is disk $s^1_{i-1}\in S$ such that $s^1_{i-1}\prec s^*$, $\hp_{i-1} \in s^1_{i-1}$, $\hp_k \notin s^1_{i-1}$, and $\hp_i \notin s^1_{i-1}$. Depending on whether $x(\hp_i) < x(\hp_{i-1})$ or $x(\hp_{i-1}) < x(\hp_i)$, there are two further subcases. 
    \begin{enumerate}
        \item If $x(\hp_i) < x(\hp_{i-1})$, then by Observation~\ref{obser:diskS_star}, we have $s^*\prec \hs_{i-1}$, and thus $s^1_{i-1}\prec \hs_{i-1}$ since $s^1_{i-1}\prec s^*$. Because $\hp_{i-1}$ is in $s^1_{i-1}$ and vertically above $\hs_{i-1}$, $s^1_{i-1}\prec \hs_{i-1}$, $x(\hp_i) < x(\hp_{i-1})$ and $\hp_i \notin s^1_{i-1}$, we have $\hp_i \notin \hs_{i-1}$ by Observation~\ref{obser:fifo10} (see Fig.~\ref{fig:prunePoint3}).
        But this incurs contradiction since $\hp_i \in \hs_{i-1}$ by the definition of $\hp_i$. 

    \begin{figure}[h]
    \begin{minipage}[t]{0.49\textwidth}
    \begin{center}
    \includegraphics[height=0.8in]{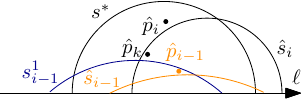}
    \caption{\footnotesize Illustrating the sub-case $x(\hp_k) < x(\hp_{i-1})$ and $x(\hp_i) < x(\hp_{i-1})$.}
    \label{fig:prunePoint3}
    \end{center}
    \end{minipage}
    \hspace{0.05in}
    \begin{minipage}[t]{0.49\textwidth}
    \begin{center}
    \includegraphics[height=0.8in]{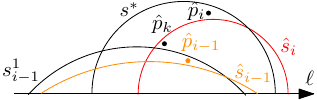}
    \caption{\footnotesize Illustrating the sub-case $x(\hp_k) < x(\hp_{i-1}) < x(\hp_i)$.}   
    \label{fig:prunePoint6}
    \end{center}
    \end{minipage}
    \vspace{-0.15in}
    \end{figure}

        \item If $x(\hp_{i-1}) < x(\hp_i)$, then by Observation~\ref{obser:diskS_star}, we have $\hs_{i-1}\prec s^*$. Since $s^*\prec \hs_i$ by Observation~\ref{obser:si}, we further have $\hs_{i-1}\prec \hs_i$. Since $\hp_{i-1}$ is in $\hs_i$ and vertically above $\hs_{i-1}$, $\hs_{i-1}\prec \hs_i$, $x(\hp_{i-1}) < x(\hp_i)$, and $\hp_i \notin \hs_i$,  we have $\hp_{i} \notin \hs_{i-1}$ by Observation~\ref{obser:fifo10} (see Fig.~\ref{fig:prunePoint6}). But this incurs contradiction since $\hp_i \in \hs_{i-1}$ by the definition of $\hp_i$.               
    \end{enumerate}


    \item 
    If $x(\hp_{i-1}) < x(\hp_k)$, 
    then by Lemma~\ref{lemma:lmPruningPoints} (with $j=i-1$), there is a disk $s^2_{i-1}\in S$ such that $s^*\prec s^2_{i-1}$, $\hp_{i-1} \in s^2_{i-1}$, and $\hp_k \notin s^2_{i-1}$; further, $\hp_i \notin s^2_{i-1}$ if $s^2_{i-1}\prec \hs_i$. 
    
    Recall that $x(\hp_k) < x(\hp_i)$. As $x(\hp_{i-1}) < x(\hp_k)$, we have $x(\hp_{i-1}) < x(\hp_i)$. By Observation~\ref{obser:diskS_star}, $\hs_{i-1}\prec s^*$, and thus $\hs_{i-1}\prec s^2_{i-1}$ since $s^*\prec s^2_{i-1}$. Also, since $x(\hp_{i-1}) < x(\hp_i)$, we already proved above (in the case $k=i-1$) that $\hp_{i-1} \notin \hs_i$ must hold.  

    Note that since $\hp_{i-1}\not\in \hs_{i}$ while $\hp_{i-1}\in s^2_{i-1}$, we have $\hs_{i}\neq s^2_{i-1}$. 
    Depending on whether $s^2_{i-1}\prec \hs_i$ or $\hs_i\prec s^2_{i-1}$, there are two sub-cases. 
    \begin{enumerate}
        \item If $s^2_{i-1}\prec \hs_i$, then $\hp_i \notin s^2_{i-1}$ (by Lemma~\ref{lemma:lmPruningPoints} as discussed above). 
        Since $\hp_{i-1}$ is in $s^2_{i-1}$ and vertically above $\hs_{i-1}$, $\hs_{i-1}\prec s^2_{i-1}$, $x(\hp_{i-1}) < x(\hp_i)$, and $\hp_i \notin s^2_{i-1}$, we have $\hp_{i} \notin \hs_{i-1}$ by Observation~\ref{obser:fifo10} (see Fig.~\ref{fig:prunePoint4}).
        But this incurs contradiction since $\hp_i \in \hs_{i-1}$ by the definition of $\hp_i$. 

        \item Suppose $\hs_i\prec s^2_{i-1}$. We claim that $\hp_{i-1}$ must be vertically above $\hs_i$. Indeed, since $\hp_{i-1} \notin \hs_i$, it suffices to show that $\hp_{i-1}$ is to the right (resp., left) of the left (resp., right) endpoint of the lower segment of $\hs_i$. On the one hand, since $x(\hp_{i-1}) < x(\hp_i)$ and $\hp_i$ is vertically above $\hs_i$, we know that $\hp_{i-1}$ is to the left of the right endpoint of the lower segment of $\hs_i$. On the other hand, since $\hp_{i-1}\in s^2_{i-1}$ (by Lemma~\ref{lemma:lmPruningPoints} as discussed above) and $\hs_i\prec s^2_{i-1}$, we obtain that $\hp_{i-1}$ is to the right of the left endpoint of the lower segment of $\hs_i$. As such, $\hp_{i-1}$ must be vertically above $\hs_i$.

        Since $\hp_{i-1}$ is in $s^2_{i-1}$ and vertically above $\hs_i$, $\hs_i\prec s^2_{i-1}$, $x(\hp_{i-1}) < x(\hp_k)$, and $\hp_k \in \hs_i$, we have $\hp_k \in s^2_{i-1}$ by Observation~\ref{obser:fifo10} (see Fig.~\ref{fig:prunePoint7}). 
        But this incurs contradiction since $\hp_k \notin s^2_{i-1}$ (by Lemma~\ref{lemma:lmPruningPoints} as discussed above). 
        
    \end{enumerate}
    \begin{figure}[H]
    \begin{minipage}[t]{0.49\textwidth}
    \begin{center}
    \includegraphics[height=0.8in]{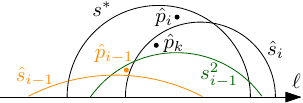}
    \caption{\footnotesize Illustrating the sub-case $x(\hp_{i-1}) < x(\hp_k)$ and $s^2_{i-1}\prec \hs_i$.}
    \label{fig:prunePoint4}
    \end{center}
    \end{minipage}
    \hspace{0.05in}
    \begin{minipage}[t]{0.49\textwidth}
    \begin{center}
    \includegraphics[height=1.0in]{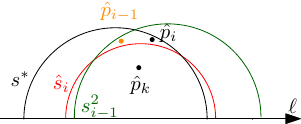}
    \caption{\footnotesize Illustrating the sub-case $x(\hp_{i-1}) < x(\hp_k)$ and $\hs_i\prec s^2_{i-1}$.}   
    \label{fig:prunePoint7}
    \end{center}
    \end{minipage}
    \vspace{-0.15in}
    \end{figure}
    
\end{enumerate}

In summary, the above obtains contradiction and thus our assumption $\hp_k\in \hs_i$ is false. This proves that $\hp_k\not\in \hs_i$. It remains to prove Lemma~\ref{lemma:lmPruningPoints}, which is done in Section~\ref{sec:lmPruningPoints}. 

\subsection{Proof of Lemma~\ref{lemma:lmPruningPoints}} 
\label{sec:lmPruningPoints}
We now prove Lemma~\ref{lemma:lmPruningPoints}. Recall that $\hp_{k}\in \hs_i$ and $x(\hp_{k}) < x(\hp_i)$ from the condition of the lemma. 
We use induction.

\paragraph{\bf Base case: $\boldsymbol{j = k+1}$.} Depending whether $x(\hp_k) < (\hp_{k+1})$ or $x(\hp_{k+1}) < x(\hp_k)$, there are two cases. 

\begin{itemize}
    \item If $x(\hp_k) < (\hp_{k+1})$, we show that $\hs_k$ can serve as $s^1_{k+1}$ in the lemma statement. To this end, below we prove that (1) $\hp_{k+1} \in \hs_i$, (2) $\hs_k \prec s^*$, (3) $\hp_{k+1} \in \hs_k$, (4) $\hp_{k} \notin \hs_k$, and (5) $\hp_i \notin \hs_k$. 

    First of all, since $x(\hp_k) < (\hp_{k+1})$, (2) $\hs_k\prec s^*$ holds by Observation~\ref{obser:diskS_star}, and thus $\hs_k\prec \hs_i$ since $s^*\prec \hs_i$ by Observation~\ref{obser:si}. Because $\hp_k$ is in $\hs_i$ and vertically above $\hs_k$, $\hs_k\prec \hs_i$, $x(\hp_k) <x(\hp_i)$, and $\hp_i \notin \hs_i$, we have (5) $\hp_i \notin \hs_k$ by Observation~\ref{obser:fifo10} (see Fig.~\ref{fig:prunePoint13}).
    On the other hand, since $x(\hp_k) <x(\hp_{k+1})$ and $\hp_{k+1} \in \hs_k$, we have (1) $\hp_{k+1} \in \hs_i$ by Observation~\ref{obser:fifo10}. In addition, (3) $\hp_{k+1} \in \hs_k$ and (4) $\hp_{k} \notin \hs_k$ follow directly from their definitions. 

    \item If $x(\hp_{k+1}) < x(\hp_k)$, we show that $\hs_k$ can serve as $s^2_{k+1}$ in the lemma statement. To this end, below we prove that (1) $s^*\prec \hs_k$, (2) $\hp_{k+1} \in \hs_k$, (3) $\hp_{k} \notin \hs_k$; and (4) $\hp_i \notin \hs_k$ if $\hs_k\prec \hs_i$. 

    Since $x(\hp_{k+1}) < x(\hp_k)$, we have (1) $s^*\prec \hs_k$ by Observation~\ref{obser:diskS_star}. (2) $\hp_{k+1} \in \hs_k$ and (3) $\hp_{k} \notin \hs_k$ follow directly from their definitions. 
    

    To prove (4), suppose $\hs_k\prec \hs_i$. Because $\hp_k$ is in $\hs_i$ and vertically above $\hs_k$, $\hs_k\prec \hs_i$, $x(\hp_k) <x(\hp_i)$, and $\hp_i \notin \hs_i$, we have (4) $\hp_i \notin \hs_k$ by Observation~\ref{obser:fifo10} (see Fig.\ref{fig:prunePoint14}). 
\end{itemize}

    \begin{figure}[H]
    \begin{minipage}[t]{0.49\textwidth}
    \begin{center}
    \includegraphics[height=0.8in]{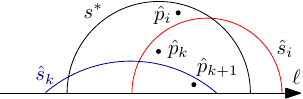}
    \caption{\footnotesize Illustrating the case $x(\hp_k) < (\hp_{k+1})$.}
    \label{fig:prunePoint13}
    \end{center}
    \end{minipage}
    \hspace{0.05in}
    \begin{minipage}[t]{0.49\textwidth}
    \begin{center}
    \includegraphics[height=0.8in]{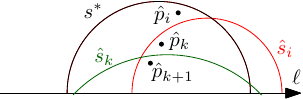}
    \caption{\footnotesize Illustrating the case $x(\hp_{k+1}) < x(\hp_k)$.}   
    \label{fig:prunePoint14}
    \end{center}
    \end{minipage}
    \vspace{-0.15in}
    \end{figure}

\paragraph{\bf Inductive step.} Assume for the inductive purpose that the lemma statement holds for $j\in [k+1,i-2]$, that is, 
    \begin{enumerate}
            \item  If $x(\hp_{k})< x(\hp_j)$, then $\hp_j \in \hs_i$ and there exists a disk $s^1_j\in S$ such that $s^1_j\prec s^*$, $\hp_j \in s^1_j$, $\hp_{k} \notin s^1_j$, and $\hp_i \notin s^1_j$.
    
            \item  If $x(\hp_j) < x(\hp_{k})$, then there exists a disk $s^2_j\in S$ such that $s^*\prec s^2_j$, $\hp_j \in s^2_j$, and $\hp_{k} \notin s^2_j$; furthermore, $\hp_i \notin s^2_j$ if $s^2_j\prec \hs_i$. 
    \end{enumerate}
    
In what follows, we show that the lemma statement also holds for $j+1$, i.e., 
   \begin{enumerate}
            \item  If $x(\hp_{k})< x(\hp_{j+1})$, then $\hp_{j+1} \in \hs_i$ and there exists a disk $s^1_{j+1}\in S$ such that $s^1_{j+1}\prec s^*$, $\hp_{j+1} \in s^1_{j+1}$, $\hp_{k} \notin s^1_{j+1}$, and $\hp_i \notin s^1_{j+1}$.
    
            \item  If $x(\hp_{j+1}) < x(\hp_{k})$, then there exists a disk $s^2_{j+1}\in S$ such that $s^*\prec s^2_{j+1}$, $\hp_{j+1} \in s^2_{j+1}$, and $\hp_{k} \notin s^2_{j+1}$; further, $\hp_i \notin s^2_{j+1}$ if $s^2_{j+1}\prec \hs_i$. 
    \end{enumerate}

Depending on whether $x(\hp_{k})< x(\hp_j)$ or $x(\hp_j) < x(\hp_{k})$, there are two cases. 

\paragraph{{\bf The case $\boldsymbol{x(\hp_k)<x(\hp_j)}$.}} In this case, according to the inductive hypothesis, $\hp_j \in \hs_i$ and there exists a disk $s^1_j\in S$ such that $s^1_j\prec s^*$, $\hp_j \in s^1_j$, $\hp_{k} \notin s^1_j$, and $\hp_i \notin s^1_j$. 
Since $s^*\prec \hs_i$ by Observation~\ref{obser:si}, we have $s^1_j\prec \hs_i$.  
Depending on whether $x(\hp_j) < x(\hp_{j+1})$ or $x(\hp_{j+1}) < x(\hp_j)$, there are two sub-cases.            
\begin{enumerate}
\item The sub-case $x(\hp_j) < x(\hp_{j+1})$.  Since $x(\hp_k)<x(\hp_j)$, we have $x(\hp_k) < x(\hp_{j+1})$. Hence, our goal is to prove $\hp_{j+1} \in \hs_i$ and find a disk to serve as $s_{j+1}^1$. We discuss the two situations depending on whether $\hp_{j+1} \in s^1_j$.
            
\begin{enumerate}        
\item \label{item:case10}
If $\hp_{j+1} \in s^1_j$, we show that $s^1_j$ can serve as $s^1_{j+1}$. To this end, below we prove that 
(1) $\hp_{j+1} \in \hs_i$, (2) $s^1_{j}\prec s^*$, (3) $\hp_{j+1} \in s^1_{j}$, (4) $\hp_{k} \notin s^1_{j}$, and (5) $\hp_i \notin s^1_{j}$.
            
First of all, (3) $\hp_{j+1} \in s^1_j$ is vacuously true. Also,   
(2) $s^1_{j}\prec s^*$, (4) $\hp_{k} \notin s^1_{j}$, and (5) $\hp_i \notin s^1_{j}$ follows directly from the inductive hypothesis as discussed above.          

                It remains to prove (1) $\hp_{j+1} \in \hs_i$.
                We claim that $\hp_k$ is vertically above $s^1_j$  (see Fig.~\ref{fig:prunePoint15}). Indeed, since $\hp_k\not\in s^1_j$, it suffices to show that $\hp_k$ is to the left (resp., right) of the right (resp., left) endpoint of the lower segment of $s^1_j$. On the one hand, since $\hp_k\in \hs_i$ and $s^1_j\prec \hs_i$, $\hp_k$ must be to the right of the left endpoint of the lower segment of $s^1_j$. On the other hand, since $x(\hp_k)<x(\hp_j)$ and $\hp_j\in s^1_j$, we can obtain that $\hp_k$ is to the left of the right endpoint of the lower segment of $s^1_j$. As such, $\hp_k$ must be vertically above $s^1_j$.
                

                Because $\hp_k$ is in $\hs_i$ and vertically above $s^1_j$, $s^1_j\prec \hs_i$, $x(\hp_k)<x(\hp_{j+1})$, and $\hp_{j+1} \in s^1_j$, we have (1) $\hp_{j+1} \in \hs_i$ by Observation~\ref{obser:fifo10}.

    \begin{figure}[H]
    \begin{minipage}[t]{0.49\textwidth}
    \begin{center}
    \includegraphics[height=1.0in]{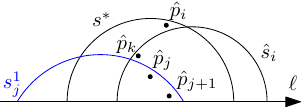}
    \caption{\footnotesize Illustrating the case $x(\hp_k)<x(\hp_j)<x(\hp_{j+1})$ and $\hp_{j+1} \in s^1_j$.}
    \label{fig:prunePoint15}
    \end{center}
    \end{minipage}
    \hspace{0.05in}
    \begin{minipage}[t]{0.49\textwidth}
    \begin{center}
    \includegraphics[height=1.0in]{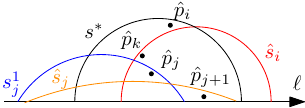}
    \caption{\footnotesize Illustrating the case  $x(\hp_k)<x(\hp_j)<x(\hp_{j+1})$ and $\hp_{j+1} \not\in s^1_j$.}   
    \label{fig:prunePoint16}
    \end{center}
    \end{minipage}
    \vspace{-0.15in}
    \end{figure}
                
                \item  If $\hp_{j+1} \notin s^1_j$, we show that $\hs_j$ can serve as $s^1_{j+1}$. To this end, below we prove that (1) $\hp_{j+1} \in \hs_i$, (2) $\hs_j\prec s^*$, (3) $\hp_{j+1} \in \hs_j$, (4) $\hp_{k} \notin \hs_j$, and (5) $\hp_i \notin \hs_j$.

                First of all, (3) $\hp_{j+1} \in \hs_j$ holds by definition. Since $x(\hp_j) < x(\hp_{j+1})$, we have (2) $\hs_j\prec s^*$ by Observation~\ref{obser:diskS_star}. Further, as $s^*\prec \hs_i$ by Observation~\ref{obser:si}, we have $\hs_j\prec \hs_i$.
               

                Note that since $\hp_{j+1}$ is in $\hs_j$ but not in $s_j^1$, $\hs_j\neq s_j^1$. 
                We claim $s^1_j\prec \hs_j$. Indeed, assume to the contrary that $\hs_j\prec s^1_j$. 
                Then, because $\hp_j$ is in $s^1_j$ and vertically above $\hs_j$, $x(\hp_j) < x(\hp_{j+1})$, and $\hp_{j+1} \notin s^1_j$, we have $\hp_{j+1} \notin \hs_j$ by 
                Observation~\ref{obser:fifo10}. But this incurs contradiction since $\hp_{j+1} \in \hs_j$ by definition. 
                As such, $s^1_j\prec \hs_j$ holds. 

                Because $\hp_j$ is in $s^1_j$ and vertically above $\hs_j$, $s^1_j\prec \hs_j$, $x(\hp_k) < x(\hp_j)$ and $\hp_k \notin s^1_j$, we have (4) $\hp_k \notin \hs_j$ by Observation~\ref{obser:fifo10} (see Fig.~\ref{fig:prunePoint16}). 

                

                Because $\hp_j$ is in $\hs_i$ and vertically above $\hs_j$, $\hs_j\prec\hs_i$, $x(\hp_j) <x(\hp_{j+1})$, and $\hp_{j+1} \in \hs_j$, we have (1) $\hp_{j+1} \in \hs_i$ by Observation~\ref{obser:fifo10} (see Fig.~\ref{fig:prunePoint16}).

                We claim that $\hp_k$ is vertically above $\hs_j$. Indeed, since $\hp_k\not\in \hs_j$ (as already proved above), it suffices to show that $\hp_k$ is to the left (resp., right) of the right (resp., left) endpoint of the lower segment of $\hs_j$. On the one hand, since $\hp_k\in \hs_i$ and $\hs_j\prec \hs_i$, $\hp_k$ must be to the right of the left endpoint of the lower segment of $\hs_j$. On the other hand, since $x(\hp_k)<x(\hp_j)$ and $\hp_j$ is vertically above $\hs_j$, we can obtain that $\hp_k$ is to the left of the right endpoint of the lower segment of $\hs_j$. As such, $\hp_k$ must be vertically above $\hs_j$.

                Because $\hp_k$ is in $\hs_i$ and vertically above $\hs_j$, $\hs_j\prec\hs_i$, $x(\hp_k) <x(\hp_i)$, and $\hp_i \notin \hs_i$, we have (5) $\hp_i \notin \hs_j$ by Observation~\ref{obser:fifo10} (see Fig.~\ref{fig:prunePoint16}). 
            \end{enumerate}

    \item The sub-case $x(\hp_{j+1}) < x(\hp_j)$. Since $x(\hp_{j+1}) < x(\hp_j)$, we have $s^*\prec \hs_j$ by Observation~\ref{obser:diskS_star}. Since $s_j^1\prec s^*$ (by the inductive hypothesis), we obtain $s_j^1\prec \hs_j$. 

    We first prove that $\hp_{j+1} \in s^1_j$. Indeed, assume to the contrary that $\hp_{j+1}\notin s^1_j$ (see Fig.~\ref{fig:prunePoint18}). Then, because $\hp_j \in s^1_j$ and $\hp_j$ is vertically above $\hs_j$, $s_j^1\prec \hs_j$, $x(\hp_{j+1}) < x(\hp_j)$, and $\hp_{j+1} \notin s^1_j$, we have $\hp_{j+1} \notin \hs_j$ by  Observation~\ref{obser:fifo10}. But this incurs contradiction since $\hp_{j+1} \in \hs_j$ by the definition of $\hp_{j+1}$. As such, $\hp_{j+1} \in s^1_j$ holds.
    

    \begin{figure}[H]
    \begin{minipage}[t]{0.49\textwidth}
    \begin{center}
    \includegraphics[height=1.1in]{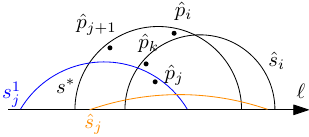}
    \caption{\footnotesize Illustrating the case $x(\hp_k)<x(\hp_j)$, $x(\hp_{j+1})<x(\hp_j)$, and $\hp_{j+1}\not\in s^1_j$.}   
    \label{fig:prunePoint18}
    \end{center}
    \end{minipage}
    \hspace{0.05in}
    \begin{minipage}[t]{0.49\textwidth}
    \begin{center}
    \includegraphics[height=1.1in]{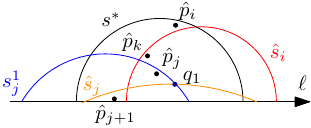}
    \caption{\footnotesize Illustrating the case $x(\hp_{j+1})<x(\hp_k)<x(\hp_j)$ and $\hp_{j+1}\in s^1_j$.}
    \label{fig:prunePoint17}
    \end{center}
    \end{minipage}
    \vspace{-0.15in}
    \end{figure}

    We discuss the two situations depending on whether $x(\hp_k) < x(\hp_{j+1})$ or $x(\hp_{j+1}) < x(\hp_k)$.
    \begin{enumerate}      
    
        \item If $x(\hp_k) < x(\hp_{j+1})$, then our goal is to prove $\hp_{j+1} \in \hs_i$ and find a disk to serve as $s^1_{j+1}$. We show that $s^1_j$ can serve as $s^1_{j+1}$. To this end, below we prove that 
        (1) $\hp_{j+1} \in \hs_i$, (2) $s^1_{j}\prec s^*$, (3) $\hp_{j+1} \in s^1_{j}$, (4) $\hp_{k} \notin s^1_{j}$, and (5) $\hp_i \notin s^1_{j}$. The proof is very similar to that for the case~\ref{item:case10}  (see Fig.~\ref{fig:prunePoint15} with $\hp_j$ and $\hp_{j+1}$ swapped). 
            
        First of all, (3) $\hp_{j+1} \in s^1_j$ has been proved above. Also,   
        (2) $s^1_{j}\prec s^*$, (4) $\hp_{k} \notin s^1_{j}$, and (5) $\hp_i \notin s^1_{j}$ follow directly from the inductive hypothesis as discussed above. 

        Using exactly the same argument as in the case~\ref{item:case10}, we can show that $\hp_k$ is vertically above $s^1_j$. Recall that $s^1_j\prec \hs_i$.
        Because $\hp_k$ is in $\hs_i$ and vertically above $s^1_j$, $s^1_j\prec \hs_i$, $x(\hp_k)<x(\hp_{j+1})$, and $\hp_{j+1} \in s^1_j$, we have (1) $\hp_{j+1} \in \hs_i$ by Observation~\ref{obser:fifo10}.
        

        \item 
        If $x(\hp_{j+1}) < x(\hp_k)$, then our goal is to find a disk to serve as $s^2_{j+1}$. We show that $\hs_j$ can serve as $s^2_{j+1}$. To this end,  below we prove that (1) $s^*\prec \hs_j$, (2) $\hp_{j+1} \in \hs_j$, (3) $\hp_{k} \notin \hs_j$; and (4) $\hp_i \notin \hs_j$ if $\hs_j\prec \hs_i$.

        We already showed above that (1) $s^*\prec \hs_j$ holds. Also,  (2) $\hp_{j+1} \in \hs_j$ follows the definition of $\hp_{j+1}$. 
        Because $\hp_j \in s^1_j$ and $\hp_j$ is vertically above $\hs_j$, $s_j^1\prec \hs_j$, $x(\hp_k) < x(\hp_j)$, and $\hp_k \notin s^1_j$, we have (3) $\hp_k \notin \hs_j$ by Observation~\ref{obser:fifo10} (see Fig.~\ref{fig:prunePoint17}).


        To prove (4), suppose $\hs_j\prec \hs_i$. We claim that $\hp_k$ is vertically above $\hs_j$ (see Fig.~\ref{fig:prunePoint17}). Indeed, since $\hp_k\not\in \hs_j$, it suffices to show that $\hp_k$ is to the left (resp., right) of the right (resp., left) endpoint of the lower segment of $\hs_j$. On the one hand, since $x(\hp_k)<x(\hp_j)$, $\hp_j\in s^1_j$, and $s^1_j\prec \hs_j$, $\hp_k$ must be to the left of the right endpoint of the lower segment of $\hs_j$. On the other hand, since $x(\hp_{j+1})<x(\hp_k)$ and $\hp_{j+1}\in \hs_j$, we can obtain that $\hp_k$ is to the right of the left endpoint of the lower segment of $\hs_j$. As such, $\hp_k$ must be vertically above $\hs_j$.
                
        Because $\hp_k$ is in $\hs_i$ and vertically above $\hs_j$, 
         $\hs_j\prec \hs_i$, $x(\hp_k) < x(\hp_i)$, and $\hp_i \notin \hs_i$, we have (4) $\hp_i \notin \hs_j$ by Observation~\ref{obser:fifo10} (see Fig.~\ref{fig:prunePoint17}). 
    \end{enumerate}

        \end{enumerate}

    \paragraph{{\bf The case $\boldsymbol{x(\hp_j) < x(\hp_k)}$.}} 
    In this case, according to the inductive hypothesis, there exists a disk $s^2_j\in S$ such that $s^*\prec s^2_j$, $\hp_j \in s^2_j$, and $\hp_{k} \notin s^2_j$; further, $\hp_i \notin s^2_j$ if $s^2_j\prec \hs_i$. Depending on whether $x(\hp_{j+1}) < x(\hp_j)$ and $x(\hp_j) < x(\hp_{j+1})$, there are two sub-cases.

    \begin{enumerate}
        \item The sub-case $x(\hp_{j+1}) < x(\hp_j)$. Since $x(\hp_j) < x(\hp_k)$, we have $x(\hp_{j+1}) < x(\hp_k)$.  Therefore, our goal is to find a disk to serve as $s^2_{j+1}$. We discuss two situations depending on whether $\hp_{j+1} \in s^2_j$. 
        \begin{enumerate}
        \item If $\hp_{j+1} \in s^2_j$, then we show that $s^2_j$ can serve as $s^2_{j+1}$ (see Fig.\ref{fig:prunePoint19}). To this end, we prove that (1) $s^*\prec s^2_{j}$, (2) $\hp_{j+1} \in s^2_{j}$, and (3) $\hp_{k} \notin s^2_{j}$; and (4) $\hp_i \notin s^2_{j}$ if $s^2_{j}\prec \hs_i$.
        
        First of all, we already have (1), (3), and (4) from the inductive hypothesis discussed above. (2) follows the condition of this case. 

    \begin{figure}[H]
    \begin{minipage}[t]{0.49\textwidth}
    \begin{center}
    \includegraphics[height=1.1in]{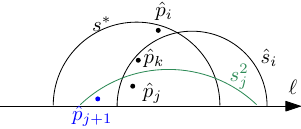}
    \caption{\footnotesize Illustrating the case $x(\hp_{j+1}) < x(\hp_j) < x(\hp_k)$ and $\hp_{j+1}\in s^2_j$.}
    \label{fig:prunePoint19}
    \end{center}
    \end{minipage}
    \hspace{0.05in}
    \begin{minipage}[t]{0.49\textwidth}
    \begin{center}
    \includegraphics[height=1.1in]{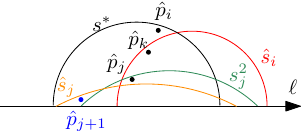}
    \caption{\footnotesize Illustrating the case $x(\hp_{j+1}) < x(\hp_j) < x(\hp_k)$ and $\hp_{j+1}\not\in s^2_j$.}   
    \label{fig:prunePoint20}
    \end{center}
    \end{minipage}
    \vspace{-0.15in}
    \end{figure}
        
        \item  If $\hp_{j+1} \notin s^2_j$, we show that $\hs_j$ can serve as $s^2_{j+1}$. To this end, we prove that (1) $s^*\prec \hs_j$, (2) $\hp_{j+1} \in \hs_j$, and (3) $\hp_{k} \notin \hs_j$; and (4) $\hp_i \notin \hs_j$ if $\hs_j\prec \hs_i$.

        First of all, since $x(\hp_{j+1}) < x(\hp_j)$, we have (1) $s^*\prec \hs_j$ by Observation~\ref{obser:diskS_star}. Also, (2) holds by definition. 
                              
        Since $\hp_j$ is in $s^2_j$ but not in $\hs_j$, $s^2_j\neq \hs_j$. We claim that $\hs_j\prec s^2_j$. Indeed, assume to the contrary that $s^2_j\prec \hs_j$. Then, since $\hp_j$ is in $s^2_j$ and vertically above $\hs_j$, $x(\hp_{j+1}) < x(\hp_j)$, and $\hp_{j+1} \notin s^2_j$, we have $\hp_{j+1}\not\in \hs_j$ by Observation~\ref{obser:fifo10}. But this incurs contradiction since $\hp_{j+1}\in \hs_j$ by the definition of $\hp_{j+1}$. As such, $\hs_j\prec s^2_j$ holds. 

        Since $\hp_j$ is in $s^2_j$ and vertically above $\hs_j$, $\hs_j\prec s^2_j$, $x(\hp_j) < x(\hp_k)$, and $\hp_k \notin s^2_j$, we have (3) $\hp_k \notin \hs_j$ by Observation~\ref{obser:fifo10} (see Fig.\ref{fig:prunePoint20}). 
                

        To prove (4), suppose $\hs_j\prec \hs_i$. We claim that $\hp_k$ is vertically above $\hs_j$ (see Fig.\ref{fig:prunePoint20}). Indeed, since $\hp_k \notin \hs_j$ as proved above, it suffices to show that $\hp_k$ is to the left (resp., right) of the right (resp., left) endpoint of the lower segment of $\hs_j$. On the one hand, since $\hp_j$ is vertically above $\hs_j$ and $x(\hp_j)<x(\hp_k)$, $\hp_k$ must be to the right of the left endpoint of the lower segment of $\hs_j$. On the other hand, since $\hp_k\in s^*$ (by Observation~\ref{obser:diskSstar}) and $s^*\prec \hs_j$, we know that $\hp_k$ must be to the left of the right endpoint of the lower segment of $\hs_j$. As such, $\hp_k$ must be vertically above $\hs_j$.

        Recall that $\hp_k \in \hs_i$. Since $\hp_k$ is in $\hs_i$ and vertically above $\hs_j$, $\hs_j\prec \hs_i$, $x(\hp_k) < x(\hp_i)$, and $\hp_i \notin \hs_i$, we have (4) $\hp_i \notin \hs_j$ by Observation~\ref{obser:fifo10}.
                
        \end{enumerate}

    \item The sub-case $x(\hp_j) < x(\hp_{j+1})$. Since $x(\hp_j) < x(\hp_{j+1})$, we have $\hs_j\prec s^*$ by Observation~\ref{obser:diskS_star}. As $s^*\prec \hs_i$ (by Observation~\ref{obser:si}) and $s^*\prec s^2_j$ (by the inductive hypothesis as discuss above), we have both $\hs_j\prec \hs_i$ and $\hs_j\prec s^2_j$. 

    We first prove that $\hp_{j+1} \in s^2_j$. Indeed, assume to the contrary that $\hp_{j+1} \not\in s^2_j$ (see Fig.~\ref{fig:prunePoint22}). Because $\hp_j$ is in $s^2_j$ and vertically above $\hs_j$, $\hs_j\prec s^2_j$, $x(\hp_j) < x(\hp_{j+1})$ and $\hp_{j+1} \notin s^2_j$, we have $\hp_{j+1}\not\in \hs_j$ by Observation~\ref{obser:fifo10}. But this incurs contradiction since $\hp_{j+1}\in \hs_j$ by definition.
            
    We discuss the two situations depending on whether $x(\hp_{j+1}) < x(\hp_k)$ or $x(\hp_k) < x(\hp_{j+1})$.
          
    \begin{enumerate}
    \item If $x(\hp_{j+1}) < x(\hp_k)$, then our goal is to find a disk to serve as $s^2_{j+1}$. We show that $s^2_j$ can serve as $s^2_{j+1}$. To this end, we prove that (1) $s^*\prec s^2_{j}$, (2) $\hp_{j+1} \in s^2_{j}$, and (3) $\hp_{k} \notin s^2_{j}$; and (4) $\hp_i \notin s^2_{j}$ if $s^2_{j}\prec \hs_i$.
        
        First of all, we already have (1), (3), and (4) from the inductive hypothesis. (2) has been proved above. 
                    
                    \item If $x(\hp_k) < x(\hp_{j+1})$, then our goal is to prove $\hp_{j+1} \in \hs_i$ and find a disk to serve as $s^1_{j+1}$. We                     
                    show that $\hs_j$ can serve as $s^1_{j+1}$. To this end, we prove that (1) $\hp_{j+1} \in \hs_i$, (2) $\hs_j\prec s^*$, (3) $\hp_{j+1} \in \hs_j$, (4) $\hp_{k} \notin \hs_j$, and (5) $\hp_i \notin \hs_j$.

                    First of all, (3) holds by definition. Also, we already proved (2) above. 
                    
                    Because $\hp_j$ is in $s^2_j$ and vertically above $\hs_j$, $\hs_j\prec s^2_j$, $x(\hp_j) < x(\hp_k)$, and $\hp_k \notin s^2_j$, we have (4) $\hp_k \notin \hs_j$ by Observation~\ref{obser:fifo10} (see Fig.~\ref{fig:prunePoint21}).                   

                    We argue that $\hp_k$ must be vertically above $\hs_j$. Indeed, since $\hp_k\not\in \hs_j$, it suffices to show that $\hp_k$ is to the left (resp., right) of the right (resp., left) endpoint of the lower segment of $\hs_j$. On the one hand, since $\hp_j$ is vertically above $\hs_j$ and $x(\hp_j)<x(\hp_k)$, $\hp_k$ must be to the right of the left endpoint of the lower segment of $\hs_j$. On the other hand, since $x(\hp_k) < x(\hp_{j+1})$ and $\hp_{j+1}\in \hs_j$, we know that $\hp_k$ must be to the left of the right endpoint of the lower segment of $\hs_j$. As such, we obtain that $\hp_k$ is vertically above $\hs_j$.

                    Because $\hp_k$ is in $\hs_i$ and vertically above $\hs_j$, $\hs_j\prec \hs_i$, $x(\hp_k) < x(\hp_i)$, and $\hp_i \notin \hs_i$, we have (5) $\hp_i \notin \hs_j$ by Observation~\ref{obser:fifo10} ( see Fig.~\ref{fig:prunePoint21}). Similarly, because $\hp_k$ is in $\hs_i$ and vertically above $\hs_j$, $\hs_j\prec \hs_i$, $x(\hp_k) < x(\hp_{j+1})$, and $\hp_{j+1} \in \hs_j$, we have (1) $\hp_{j+1} \in \hs_i$ by Observation~\ref{obser:fifo10}.
            \end{enumerate}

    \begin{figure}[H]
    \begin{minipage}[t]{0.49\textwidth}
    \begin{center}
    \includegraphics[height=1.1in]{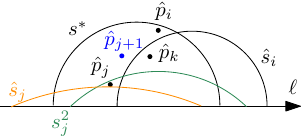}
    \caption{\footnotesize Illustrating the case $x(\hp_j) < x(\hp_k)$, $x(\hp_j)<x(\hp_{j+1})$, and $\hp_{j+1}\not\in s^2_j$.}   
    \label{fig:prunePoint22}
    \end{center}
    \end{minipage}
    \hspace{0.05in}
    \begin{minipage}[t]{0.49\textwidth}
    \begin{center}
    \includegraphics[height=1.1in]{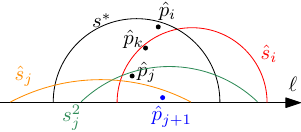}
    \caption{\footnotesize Illustrating the case $x(\hp_j) < x(\hp_k)<x(\hp_{j+1})$.}
    \label{fig:prunePoint21}
    \end{center}
    \end{minipage}
    \vspace{-0.15in}
    \end{figure}
            
        \end{enumerate}
         
In summary, the above proves that the statement of Lemma~\ref{lemma:lmPruningPoints} still holds for $j+1$. This also proves Lemma~\ref{lemma:lmPruningPoints}.

\section{Algorithm implementation}
\label{sec:hitimplement}
In this section, we present the implementation of our algorithm. In particular, we describe how to implement the first two steps of the algorithm: (1) Compute $a(i)$ and $b(i)$ for all disks $s_i\in S$; (2) find the subset $Q$ of all prunable points from $P$. 


The following lemma gives the implementation for the first step of the algorithm. 

\begin{lemma}
    \label{lem:40}
Computing $a(i)$ and $b(i)$ for all disks $s_i\in S$ can be done in $O(m\log^2n+(m+n)\log (m+n))$ time.
\end{lemma}
\begin{proof}
We only discuss how to compute $a(i)$ since computing $b(i)$ can be done analogously. 

Recall that points of $P$ are indexed in ascending order of their $x$-coordinates as $p_1,p_2,\ldots,p_n$. 
Let $T$ be a complete binary search tree whose leaves from left to right correspond to points of $P$ in their index order. Since $n=|P|$, the height of $T$ is $O(\log n)$. For each node $v\in T$, let $P_v$ denote the subset of points of $P$ in the leaves of the subtree rooted at $v$. Our algorithm is based on the following observation: a disk $s_i\in S$ contains a point of $P_v$ if and only if $s_i$ contains the closest point of $P_v$ to $c_i$, where $c_i$ is the center of $s_i$. In light of this observation, we construct the Voronoi diagram for $P_v$, denoted by $\vd_v$, and build a point location data structure on $\vd_v$ so that each point location query can be answered in $O(\log n)$ time~\cite{ref:EdelsbrunnerOp86,ref:KirkpatrickOp83}. 
For time analysis, after $\vd_v$ is computed, building the point location data structure on $\vd_v$ takes $O(|P_v|)$ time~\cite{ref:EdelsbrunnerOp86,ref:KirkpatrickOp83}. To compute $\vd_v$, if we do so from scratch, then it takes $O(|P_v|\log|P_v|)$ time. However, using Kirkpatrick's algorithm~\cite{ref:KirkpatrickEf79},  we can compute $\vd_v$ in $O(|P_v|)$ time by merging the Voronoi diagrams $\vd_u$ and $\vd_w$ for the two children $u$ and $w$ of $v$, since $P_v=P_u\cup P_w$. As such, if we construct the Voronoi diagrams for all nodes of $T$ in a bottom-up manner, the total time is linear in $\sum_{v\in T}|P_v|$, which is $O(n\log n)$. 

For each disk $s_i\in S$, we can compute $a(i)$ using $T$, as follows. 
Starting from the root of $T$, for each node $v$, we do the following. Let $u$ and $w$ be the left and right children of $v$, respectively. First, we determine whether $s_i$ contains a point of $P_u$. To this end, using a point location query on $\vd_u$, we find the point $p$ of $P_v$ closest to $c_i$. As discussed above, $s_i$ contains a point of $P_u$ if and only if $p\in s_i$. If $p\in s_i$, then we proceed with $v=u$; otherwise, we proceed with $v=w$. In this way, $a(i)$ can be computed after searching a root-to-leaf path of $T$, which has $O(\log n)$ nodes as the height of $T$ is $O(\log n)$. Because we spend $O(\log n)$ time on each node, the total time to compute $a(i)$ is $O(\log^2 n)$. The time for computing $a(i)$ for all disks $s_i\in S$ is thus $O(m\log^2 n)$. 

In summary, the overall time to compute $a(i)$ for all disks $s_i \in S$ is bounded by $O(m\log^2n+(n+m)\log (n+m))$.
\qed
\end{proof}

With $a(i)$ and $b(i)$ computed in Lemma~\ref{lem:40}, Lemma~\ref{lem:50} finds all prunable points of $P$. The algorithm of Lemma~\ref{lem:50} relies on the following observation. 

\begin{observation}\label{obser:70}
For any point $p_k\in P$, $p_k$ is prunable if and only if there is a disk $s_i\in S$ such that $p_k\not\in s_i$ and $a(i)\leq k\leq b(i)$. 
\end{observation}
\begin{proof}
If $p_k$ is prunable, then by definition there is a disk $s_i\in S$ such that $p_k\not\in s_i$ and $a(i)< k< b(i)$. 

On the other hand, suppose that there is a disk $s_i\in S$ such that $p_k\not\in s_i$ and $a(i)\leq k\leq b(i)$. By the 
definition of $a_i$, $s_i$ contains the point $p_{a(i)}$. Since $p_k\not\in s_i$, we obtain  $k\neq a(i)$. By a similar argument, we have $k\neq b(i)$. As such, since $a(i)\leq k\leq b(i)$, we can derive $a(i)< k< b(i)$. Therefore, $p_k$ is prunable. \qed
\end{proof}

\begin{lemma}
    \label{lem:50}
All prunable points of $P$ can be found in $O((n+m)\log(n+m))$ time.
\end{lemma}
\begin{proof}
Recall that $\ell$ denotes the $x$-axis. We define $T$ as the standard segment tree~\cite[Section 10.3]{ref:deBergCo08} on the $n$ points of $\ell$ whose $x$-coordinates are equal to $1,2,\ldots,n$, respectively. The height of $T$ is $O(\log n)$. For each disk $s_i\in S$, let $I_i$ denote the interval $[a(i),b(i)]$ of $\ell$. Following the definition of the standard segment tree~\cite[Section 10.3]{ref:deBergCo08}, we store $I_i$ in $O(\log n)$ nodes of $T$. For each node $v\in T$, let $S_v$ denote the subset of disks $s_i$ of $S$ whose interval $I_i$ is stored at $v$. As such, $\sum_{v\in T}|S_v|=O(m\log n)$. 

Consider a point $p_k\in P$. The tree $T$ has a leaf corresponding to a point of $\ell$ whose $x$-coordinate is equal to $k$, called {\em leaf-$k$}. Let $\pi_k$ denote the path of $T$ from the root to leaf-$k$. Following the definition of the segment tree, we have the following observation: $\bigcup_{v\in \pi_k}S_v=\{s_i\ |\ s_i\in S, a(i)\leq k\leq b(i)\}$. By Observation~\ref{obser:70}, to determine whether $p_k$ is prunable, it suffices to determine whether there is a node $v\in \pi_k$ such that $S_v$ has a disk $s_i$ that does not contain $p_k$. Recall that all points of $P$ are above $\ell$ while the centers of all disks of $S$ are below $\ell$. Let $C_v$ denote the common intersection of all disks of $S_v$ in the halfplane above $\ell$. Observe that $S_v$ has a disk $s_i$ that does not contain $p_k$ if and only if $p_k$ is outside $C_v$. Based on this observation, for each node $v\in T$, we compute $C_v$ and store it at $v$. Due to the single-intersection property that the upper arcs of every two disks of $S$ intersect at most once, $C_v$ has $O(|S_v|)$ vertices; in addition, by adapting Graham's scan, $C_v$ can be computed in $O(|S_v|)$ time if the centers of all the disks of $S_v$ are sorted by $x$-coordinate (due to the non-containment property, this is also the order of the disks sorted by the left or right endpoints of their upper arcs). Assuming that the sorted lists of $S_v$ as above are available for all nodes $v\in T$, the total time for constructing $C_v$ for all nodes $v\in T$ is linear in $\sum_{v\in T}|S_v|$, which is $O(m\log n)$. We show that the sorted lists of $S_v$ for all nodes $v\in T$ can be computed in $O(m\log m+m\log n)$ time, as follows.  At the start of the algorithm, we sort all disks of $S$ by the $x$-coordinates of their centers in $O(m\log m)$ time. Then, for each disk $s_i$ of $S$ following this sorted order, we find the nodes $v$ of $T$ where the interval $I_i$ should be stored, and add $s_i$ to $S_v$, which can be done in $O(\log n)$ time~\cite[Section 7.4]{ref:deBergCo08}. In this way, after all disks of $S$ are processed as above, $S_v$ for every node $v\in T$ is automatically sorted. As such, all processing work on $T$ together takes $O((m+n)\log (m+n))$ time. 

For each point $p_k\in P$, to determine whether $p_k$ is prunable, following the above discussion, we determine whether $p_k$ is outside $C_v$ for each node $v\in \pi_k$. Deciding whether $p_k$ is outside $C_v$ can be done in $O(\log m)$ time by binary search. Indeed, since the centers of all disks are below $\ell$, the boundary of $C_v$ consists of a segment on $\ell$ bounding $C_v$ from below and an $x$-monotone curve bounding $C_v$ from above. The projections of the vertices of $C_v$ onto $\ell$ partition $\ell$ into a set $\calI_v$ of $O(|S_v|)$ intervals. If we know the interval of $\calI_v$ that contains $x(p_k)$, the $x$-coordinate of $p_k$, then whether $p_k$ is outside $C_v$ can be determined in $O(1)$ time. Clearly, we can find the interval of $\calI_v$ that contains $x(p_k)$ in $O(\log m)$ time by binary search. In this way, whether $p_k$ is prunable can be determined in $O(\log m\log n)$ time as $\pi_k$ has $O(\log n)$ nodes. 
The time can be improved to $O(\log m+\log n)$ using fractional cascading~\cite{ref:ChazelleFr86}, as follows. 

We construct a fractional cascading data structure on the intervals of $\calI_v$ of all nodes $v\in T$, which takes $O(m\log n)$ time~\cite{ref:ChazelleFr86} since the total number of such intervals is $O(m\log n)$. With the fractional cascading data structure, for each point $p_k\in P$, we only need to do binary search on the set of the intervals stored at the root of $T$ to find the interval containing $x(p_k)$, which takes $O(\log (m\log n))$ time. Subsequently, following the path $\pi_k$ in a top-down manner, the interval of $\calI_v$ containing $x(p_k)$ for each node $v\in \pi_k$ can be determined in $O(1)$ time~\cite{ref:ChazelleFr86}. As such, whether $p_k$ is prunable can be determined in $O(\log n+\log m)$ time. Hence, the total time for checking all the points $p_k\in P$ is $O(n\log (m+n))$. 

In summary, the time complexity of the overall algorithm for finding all prunable disks of $S$ is bounded by $O((n+m)\log (n+m))$. \qed
\end{proof}

With Lemmas~\ref{lem:40} and \ref{lem:50}, Theorem~\ref{theo:hit} is proved. 

\paragraph{\bf An algebraic decision tree algorithm.} In the algebraic decision tree model, where only comparisons are counted towards time complexities, the problem can be solved in $O((n+m)\log (n+m))$ time, i.e., using $O((n+m)\log (n+m))$ comparisons. To this end, observe that the entire algorithm, with the exception of Lemma~\ref{lem:40}, takes $O((n+m)\log(n+m))$ time. As such, we only need to show that Lemma~\ref{lem:40} can be solved using $O((n+m)\log (n+m))$ comparisons. For this, notice that the factor $O(m\log^2 n)$ in the algorithm of Lemma~\ref{lem:40} is caused by the point location queries on the Voronoi diagrams $\vd_v$. The number of point location queries is $O(m\log n)$. The total combinatorial complexity of the Voronoi diagrams $\vd_v$ of all nodes $v\in T$ is $O(n\log n)$. To answer these point location queries, 
we employ a method recently introduced by Chan and Zheng~\cite{ref:ChanHo23}. In particular, by applying \cite[Theorem~7.2]{ref:ChanHo23}, all point location queries can be solved using $O((n+m)\log (n+m))$ comparisons (specifically, following the notation in \cite[Theorem~7.2]{ref:ChanHo23}, we have $t=O(n)$, $L=O(n\log n)$, $M=O(m\log n)$, and $N=O(n+m)$ in our problem; according to the theorem, all point location queries can be answered using $O(L+M+N\log N)$ comparisons, which is $O((n+m)\log (n+m))$). 

\paragraph{\bf The unit-disk case.} We consider the case where all disks of $S$ have the same radius (and the points of $P$ are separated from the centers of all disks of $S$ by the $x$-axis $\ell$). 
As discussed in Section~\ref{sec:intro}, this problem can be solved in $O((n+m)\log(n+m))$ time by reducing it to a line-separable unit-disk coverage problem and then applying the algorithm in~\cite{ref:LiuOn23}. Here, we show that our approach can provide an alternative algorithm with asymptotically the same runtime. 

We apply the same algorithm as above. Observe that the algorithm, except for Lemma~\ref{lem:40}, runs in $O((n+m)\log(n+m))$ time. Hence, it suffices to show that Lemma~\ref{lem:40} can be implemented in $O((n+m)\log(n+m))$ time for the unit-disk case. To this end, we have the following lemma. 

\begin{lemma}\label{lem:60}
If all disks of $S$ have the same radius, then $a(i)$ and $b(i)$ for all disks $s_i\in S$ can be computed in $O((n+m)\log(n+m))$ time. 
\end{lemma}
\begin{proof}
We only discuss how to compute $a(i)$ since the algorithm for $b(i)$ is similar. We modify the algorithm in the proof of Lemma~\ref{lem:40} and follow the notation there. 

For any disk $s_i\in S$, to compute $a(i)$, recall that a key subproblem is to determine whether $s_i$ contains a point of $P_v$ for a node $v\in T$. To solve the subproblem, the algorithm of Lemma~\ref{lem:40} uses Voronoi diagrams. Here, we use a different approach by exploring the property that all disks of $S$ have the same radius, say $r$. For any point $p\in P$, let $D_p$ denote the disk of radius $r$ and centered at $p$. Define $\calD_v=\{D_p\ |\ p\in P_v\}$. 
For each point $p\in P$, since $p$ is above the axis $\ell$, the portion of the boundary of $D_p$ below $\ell$ is an arc on the lower half circle of the boundary of $D_p$, and we call it the {\em lower arc} of $D_p$. 
Let $\calL_v$ denote the lower envelope of $\ell$ and the lower arcs of all disks of $\calD_v$. Our method is based on the observation that $s_i$ contains a point of $P_v$ if and only if $c_i$ is above $\calL_v$, where $c_i$ is the center of $s_i$. 

In light of the above discussion, we construct $\calL_v$ for every node $v\in T$. 
Since all disks of $\calD_v$ have the same radius and all their centers are above $\ell$, the lower arcs of every two disks of $\calD_v$ intersect at most once. Due to this single-intersection property, $\calL_v$ has at most $O(|P_v|)$ vertices. To see this, we can view the lower envelope of each lower arc of $\calD_v$ and $\ell$ as an extended arc. Every two such extended arcs still cross each other at most once and therefore their lower envelope has $O(|P_v|)$ vertices following the standard Davenport-Schinzel sequence argument~\cite{ref:SharirDa96} (see also \cite[Lemma 3]{ref:ChanAl16} for a similar problem). Notice that $\calL_v$ is exactly the lower envelope of these extended arcs and thus $\calL_v$ has $O(|P_v|)$ vertices. Note also that $\calL_v$ is $x$-monotone.  In addition, given $\calL_u$ and $\calL_w$, where $u$ and $w$ are the two children of $v$, $\calL_v$ can be computed in $O(|P_v|)$ time by a straightforward line sweeping algorithm. As such, if we compute $\calL_v$ for all nodes $v\in T$ in a bottom-up manner, the total time is linear in $\sum_{v\in T}|P_v|$, which is $O(n\log n)$. 

For each disk $s_i\in S$, we now compute $a(i)$ using $T$, as follows. Starting from the root of $T$, for each node $v$, we do the following. Let $u$ and $w$ be the left and right children of $v$, respectively. We first determine whether $c_i$ is above $\calL_u$; since $|P_u|\leq n$, this can be done in $O(\log n)$ time by binary search. More specifically, the projections of the vertices of $\calL_u$ onto $\ell$ partition $\ell$ into a set $\calI_u$ of $O(P_u)$ intervals. If we know the interval of $\calI_u$ that contains $x(c_i)$, the $x$-coordinate of $c_i$, then whether $c_i$ is above $\calL_u$ can be determined in $O(1)$ time. Clearly, finding the interval of $\calI_u$ containing $x(c_i)$ can be done by binary search in $O(\log n)$ time. If $c_i$ is above $\calL_u$, then $s_i$ must contain a point of $P_u$; in this case, we proceed with $v=u$. Otherwise, we proceed with $v=w$. In this way, $a(i)$ can be computed after searching a root-to-leaf path of $T$, which has $O(\log n)$ nodes as the height of $T$ is $O(\log n)$. Because we spend $O(\log n)$ time on each node, the total time for computing $a(i)$ is $O(\log^2 n)$. As in Lemma~\ref{lem:50}, the time can be improved to $O(\log n)$ using fractional cascading~\cite{ref:ChazelleFr86}, as follows. 

We construct a fractional cascading data structure on the intervals of $\calI_v$ of all nodes $v\in T$, which takes $O(n\log n)$ time~\cite{ref:ChazelleFr86} since the total number of such intervals is $O(n\log n)$. With the fractional cascading data structure, for each disk $s_i\in S$, we only need to do binary search on the set of the intervals stored at the root of $T$ to find the interval containing $x(c_i)$, which takes $O(\log n)$ time. After that, the interval of $\calI_u$ containing $x(c_i)$ for each node $u$ in the algorithm as discussed above can be determined in $O(1)$ time~\cite{ref:ChazelleFr86}. As such,
$a(i)$ can be computed in $O(\log n)$ time. The total time for computing $a(i)$ for all disks $s_i\in S$ is $O(m\log n)$. 

In summary, computing $a(i)$ for all disks $s_i\in S$ takes $O((n+m)\log (n+m))$ time in total. 
\qed
\end{proof}

Combining Lemma~\ref{lem:50} and \ref{lem:60} leads to an algorithm of $O((n+m)\log (n+m))$ time for the line-separable unit-disk case.

\bibliographystyle{plainurl}

\end{document}